\documentclass[a4paper,11pt]{article}

\usepackage[T1]{fontenc}
\usepackage[utf8]{inputenc}
\usepackage{lmodern}
\usepackage{amsmath,amssymb,amsthm}
\usepackage[margin=2.5cm]{geometry}
\usepackage{xurl}
\usepackage{float}
\usepackage{placeins}
\usepackage{hyperref}
\usepackage{cleveref}
\usepackage{booktabs}
\usepackage{array}
\usepackage{tikz}
\usetikzlibrary{arrows.meta,positioning}
\usepackage{xspace}
\usepackage{upquote}
\hypersetup{
  hidelinks,
  pdftitle={Certified Qualitative Analysis of the SIR ODE and Reusable Scalar Lemmas in Isabelle/HOL},
  pdfauthor={David B. Hulak, Arthur F. Ramos, Ruy J. G. B. de Queiroz}
}

\newtheorem{theorem}{Theorem}[section]
\newtheorem{lemma}[theorem]{Lemma}
\newtheorem{corollary}[theorem]{Corollary}
\newtheorem{definition}[theorem]{Definition}

\newcommand{\isa}[1]{\textsf{#1}}
\newcommand{\isatt}[1]{\texttt{#1}}
\newcommand{\Isabelle}{Isabelle/HOL\xspace}

\newcommand{\doi}[1]{\href{https://doi.org/#1}{doi:\nolinkurl{#1}}}
\newcommand{\lastcheckeddate}{3 May 2026}

\newcommand{\R}{\mathbb{R}}
\newcommand{\deriv}[2]{\frac{d#1}{d#2}}
\newcommand{\Rrecinit}{R_{\mathrm{rec},0}}
\newcommand{\Rinit}{\mathcal{R}_{\mathrm{init}}}
\newcommand{\Reff}{\mathcal{R}_{\mathrm{eff}}}
\newcommand{\Rclassical}{\mathcal{R}_{0}^{\mathrm{DFE}}}

\title{Certified Qualitative Analysis of the SIR ODE\\
and Reusable Scalar Lemmas in Isabelle/HOL}
\author{
David B. Hulak\\
\small Independent Researcher\\
\small \href{mailto:dbhulak@gmail.com}{dbhulak@gmail.com}\\
\small ORCID: \href{https://orcid.org/0009-0002-8056-1774}{0009-0002-8056-1774}
\and
Arthur F. Ramos\\
\small Microsoft\\
\small \href{mailto:arfreita@microsoft.com}{arfreita@microsoft.com}\\
\small ORCID: \href{https://orcid.org/0009-0003-3568-0325}{0009-0003-3568-0325}
\and
Ruy J. G. B. de Queiroz\\
\small Centro de Inform{\'a}tica, Universidade Federal de Pernambuco\\
\small \href{mailto:ruy@cin.ufpe.br}{ruy@cin.ufpe.br}\\
\small ORCID: \href{https://orcid.org/0000-0003-1482-0977}{0000-0003-1482-0977}
}
\date{}

\begin{document}

\maketitle

\begin{abstract}
We present a mechanically checked Isabelle/HOL bridge from the
Picard--Lindel\"of flow infrastructure in the Archive of Formal Proofs (AFP) to
selected standard qualitative facts for the mass-action, closed-population
classical susceptible--infectious--recovered (SIR) epidemic ODE.  The
epidemiological facts are classical; the contribution is reusable theorem
infrastructure connecting the AFP local-flow
construction to global forward existence, forward invariance of the nonnegative
orthant (all three compartments nonnegative), conservation, monotonicity, the
Kermack--McKendrick conserved phase-plane relation, compartment bounds, and
threshold-ratio conditions governing pointwise infectious growth and
interval-wide monotonicity.

The finite-interval qualitative facts are proved first on bounded intervals.
After the compact-continuation argument establishes global forward existence,
those facts are transferred to the unique AFP forward flow on arbitrary finite
intervals $[0,b]$ with $b>0$.  We do not
formalize stability, final-size, or asymptotic theory.  For the forward
initial-value problem in this unnormalized mass-action convention,
$\beta S(0)/\gamma$ is the initial effective threshold ratio governing initial
infectious growth when $I(0)>0$; it need not equal the fully susceptible
disease-free-equilibrium quantity.

The checked proof first establishes sign and conservation facts for the local
AFP flow, then uses the conserved nonnegative simplex---a compact set containing
each local-flow segment whose times are already known to exist---as the
compactness witness for extending to all forward times.  Thus the qualitative
facts, under their stated side conditions, apply to the unique Isabelle/AFP SIR
flow itself, not to an informally assumed trajectory.
This packages selected textbook SIR qualitative reasoning as importable Isabelle
theorem infrastructure.

The reusable layer is an Isabelle library for homogeneous-linear scalar
compartment arguments---equations of the form $X'(t)=f(t)X(t)$---together with
derivative-sign monotonicity, three-compartment conservation, and an SIR transfer
bridge to the AFP flow infrastructure.  The accompanying Isabelle artifact
builds with Isabelle~2024 and AFP~2024, and its checked source theory files
contain no \isa{sorry} proof placeholders and no \isa{oops} proof-abort commands.

\end{abstract}


\section{Introduction}
\label{sec:introduction}

Compartmental models are a cornerstone of mathematical
epidemiology~\cite{kermack1927sir,anderson1991infectious,diekmann2000mathematical,diekmann2013tools,brauer2019epidemiology,hethcote2000mathematics,murray2002mathematical}.
These models partition a population into distinct epidemiological compartments
whose sizes evolve according to systems of ordinary differential equations
(ODEs).  Their qualitative analysis is non-numerical structural reasoning about
trajectories: conservation laws, monotonicity, threshold conditions, forward
invariance, and phase-plane structure such as the Kermack--McKendrick conserved
level curves.  These facts support understanding epidemic dynamics without
closed-form solutions.  The concrete case is the closed-population
susceptible--infectious--recovered (SIR) system
$S'=-\beta SI$, $I'=\beta SI-\gamma I$, and $R'=\gamma I$, where
$S$, $I$, and $R$ denote susceptible, infectious, and recovered compartments;
the standard notation is made precise in \Cref{sec:sir-model}.
One such qualitative fact is the conserved Kermack--McKendrick level-curve
relation in the $(S,I)$ phase plane; \Cref{sec:phase-plane} gives the precise
finite-interval statement and positivity condition.
Making these qualitative facts machine-checkable matters for downstream
formal epidemiology: for example, a later verified numerical-bound argument or
model-extension proof can import sign, conservation, and threshold side
conditions instead of assuming them informally.

We present a mechanically checked Isabelle/HOL formalization of these qualitative
SIR facts, organized around reusable scalar lemmas for compartmental models.
Here the certification claim is specific: the artifact builds under
Isabelle~2024 and AFP~2024, and each displayed theorem faithfully summarizes a
checked Isabelle statement.
The novelty is formal-methods infrastructure---the checked AFP-flow bridge and
reusable interfaces---rather than new epidemiological mathematics.

At the architectural level, a central mechanization pitfall is circularity:
all-forward-time invariance cannot be a premise for proving all-forward-time
existence.
In particular, one cannot assume that a trajectory remains for all future times
in a compact nonnegative simplex before the continuation argument has established
that those future times exist.
Only solution prefixes whose times already lie in the maximal existence interval
are available for the compact-continuation argument.
\Cref{sec:non-circular-globalization} gives the proof order; sign and
conservation facts on existing local prefixes provide the compact-continuation
witness: a compact set containing the existing local trajectory, used by the AFP
continuation theorem to extend the local flow.

\begin{definition}[Forward invariance]
In this paper, a set $K\subseteq\R^3$ of SIR states is \emph{forward invariant}
for the class of differentiable triples satisfying the SIR ODE (stated above;
formalized in \Cref{sec:sir-model}) on $[a,b]$, when every solution whose state triple
$(S(a),I(a),R(a))$ lies in $K$ at its left endpoint has its state triple
$(S(t),I(t),R(t))$ in $K$ for every $t\in[a,b]$.
\end{definition}
This definition fixes terminology used throughout.
In the finite-interval statements below, ``starts'' means at the left endpoint
$a$; the flow bridge specializes this to initial time~$0$.
For ordinary mathematical reading, this means differentiable SIR trajectories
satisfying the ODE on the interval.
In the main application, $K$ is the nonnegative orthant; the formal theorem is
proved in the \isa{SIR\_solution} locale, a named, parameterized proof context
(\Cref{sec:locales}) whose assumptions include nonnegative initial data.

Despite their ubiquity, the qualitative properties of compartmental models are
often presented at a mathematical level that suppresses mechanization-relevant
hypotheses.  In textbook presentations of the SIR
model~\cite{kermack1927sir,hethcote2000mathematics,brauer2012mathematical,brauer2019epidemiology},
nonnegativity is usually handled as part of the biologically meaningful state
space or by a boundary argument, while existence, uniqueness, and
differentiability hypotheses are delegated to standard ODE theory or left in the
background.

In this development, mechanization forces these assumptions into the open: we
prove nonnegativity from the SIR ODE and nonnegative initial conditions, then
connect the result to a Picard--Lindel\"of construction of the SIR flow.  The
compactness argument may use only invariants already proved on existing
local-flow prefixes.  Because SIR-family models are building blocks for richer
formal epidemiology, even qualitative proof obligations such as nonnegativity,
conservation, and threshold monotonicity are worth making explicit.  For
example, a threshold statement is only biologically meaningful for trajectories
that remain in the nonnegative orthant (the set where all compartment values are
nonnegative).

For readers interested in formalized reasoning, the reusable architecture has
three layers: scalar calculus lemmas, SIR-specific locale instantiation, and AFP
flow globalization.  The proof-engineering difficulty is preserving theorem
side conditions across these layers, not redoing the textbook algebra by hand.
Concretely, the checked bridge must align endpoint derivative conventions,
within-interval and full-neighbourhood derivative facts, theorem-side-condition
traceability, and compact continuation without assuming the global invariants
being proved.
\Cref{sec:using-theorem-package} gives a concrete downstream import scenario.

\paragraph{Contributions.}
\begin{enumerate}
\item A \emph{reusable} \Isabelle lemma library for homogeneous-linear scalar
  compartmental ODE arguments, equations of the form $X'(t)=f(t)X(t)$, and
  three-compartment conservation (\Cref{sec:framework}), providing:
  \begin{itemize}
  \item Generic integrating-factor representation for scalar equations
    $X'=f(t)X$ with continuous coefficient \(f\).
  \item Scalar sign preservation: nonnegative initial value implies nonnegative trajectory.
  \item Interface wrappers for derivative-sign monotonicity.
  \item A three-compartment conservation theorem.
  \end{itemize}
\item An \emph{instantiation} for the classical SIR model
  (\Cref{sec:sir-model}), proving:
  \begin{itemize}
  \item Instantiation of the AFP Picard--Lindel\"of/flow infrastructure for the
    SIR vector field.
  \item Global forward existence for nonnegative initial data.
  \item Uniqueness of the forward trajectory generated by the SIR vector field.
  \item Forward invariance of the nonnegative orthant from initial conditions.
  \item Integrating-factor integral representations for $S$ and $I$.
  \item Conservation of total population ($S + I + R = N$).
  \item Monotonicity of $S$ (nonincreasing) and $R$ (nondecreasing).
  \item The Kermack--McKendrick phase-plane invariant: a conserved level-curve
    relation in the $(S,I)$ plane (under $S(a) > 0$).
  \item The initial effective threshold ratio (\Cref{sec:R0}),
    $\Rinit = \beta S(a)/\gamma$, with pointwise growth/decline conditions,
    distinct from the fully susceptible disease-free equilibrium (DFE) basic
    reproduction quantity.
  \item Monotonicity of $I$ (nonincreasing) on $[a,b]$ when $\Rinit \le 1$.
  \item Stationary infection condition: when $I(t)>0$,
    $I'(t) = 0 \iff S(t) = \gamma/\beta$.
  \item Boundedness ($0 \le X(t) \le N$ for all compartments, from
    nonnegative initial data).
  \end{itemize}
\item A non-circular proof architecture from assumptions to solutions.  The
  scalar \isa{SIR\_solution} locale records the qualitative theorems for any
  differentiable solution on a finite interval.  The \isa{SIR\_ODE} locale
  proves sign and conservation facts for the local AFP flow, uses compact
  continuation to derive the unique global forward solution from $\beta > 0$,
  $\gamma > 0$, and nonnegative initial data, and interprets (i.e.,
  instantiates) the scalar locale on every interval $[0,b]$.  This avoids the
  concrete circularity pitfall of using global nonnegativity as a premise for
  the compact witness that is itself needed to justify global existence.
\end{enumerate}

\paragraph{Approach.}
The development is structured as a locale-based framework.  The framework
theory proves results for arbitrary differentiable functions, and the SIR locale
specializes them once its ODE hypotheses match the generic theorem assumptions.
The $S$ and $I$ equations are homogeneous linear in their own compartments, so
the integrating-factor method yields sign-preserving integral representations;
the $R$ equation is handled by derivative-sign monotonicity.  The separate
existence theory then uses the AFP ODE library to obtain the unique SIR flow and
transfer the qualitative theorems to it.  Globalizing the flow means proving
that the maximal AFP existence interval contains every nonnegative time, not
assuming a global solution in advance.

\paragraph{Outline.}
\Cref{sec:background} reviews the Isabelle interfaces used by the proofs,
including the derivative endpoint convention central to the framework's
interface choices.
\Cref{sec:framework} presents the reusable lemmas, centered on the
integrating-factor theorem. \Cref{sec:sir-model} instantiates them for SIR and
connects them to the AFP flow; its SIR results include forward invariance,
the Kermack--McKendrick invariant, threshold theorems, and boundedness.
\Cref{sec:related-work} positions the result
against adjacent theorem-proving and formal-methods work. \Cref{sec:discussion}
covers assumptions, proof-engineering observations, limitations, artifact
details (\Cref{sec:artifact}), and future work, and \Cref{sec:conclusion} closes
with the main takeaway.
Throughout, displayed theorem statements are faithful prose summaries of checked
Isabelle facts and are written to include their mathematical hypotheses.
Isabelle-facing identifiers are set in sans-serif type, as in \isa{SIR\_solution},
and identify the corresponding artifact constants or locales; monospace names
such as \isatt{Compartmental\_Model} denote source-level theory or session names.

\section{Background}
\label{sec:background}

\subsection{Isabelle/HOL and HOL-Analysis}

Isabelle is a generic interactive theorem prover; its most widely used object
logic, Higher-Order Logic (HOL), is the basis of Isabelle/HOL~\cite{nipkow2002isabelle}.
The system provides a rich type system, a powerful simplifier, and the
structured Isar proof language~\cite{wenzel2002isar}, which supports proofs
inspectable at a higher level than tactic traces alone.  In this development,
Isar and locales make the reusable calculus lemmas visible as theorem contexts
rather than hidden proof scripts.
HOL functions are total, meaning defined on every input of their type.
Consequently, an interval solution is represented by a total function; the ODE
and qualitative properties are asserted only on the interval under discussion.
This representation is one reason the endpoint derivative convention below must
be stated explicitly: the function exists everywhere, while the theorem
hypotheses and conclusions are interval-scoped.

Typographically, ordinary mathematical notation such as $S(t)$ and $I(t)$
states paper-level functions and quantities, while sans-serif names such as
\isa{has\_real\_derivative} mark checked Isabelle constants and theorem names,
and typewriter names such as \isatt{HOL-Analysis} mark theory names or syntax
excerpts.

In Isabelle~2024, \isatt{HOL-Analysis} provides the real and multivariate
analysis
infrastructure used throughout the development: topological and normed vector
spaces, Fr\'echet/field/real derivatives, integration on intervals,
monotonicity consequences of the Mean Value Theorem, and the standard
transcendental functions \(\exp\) and \(\ln\)~\cite{holzl2013typeclasses}.
Its type-class hierarchy lets
the same library lemmas apply uniformly across reals, Euclidean spaces, and
other compatible structures.  The next two subsections spell out the particular
derivative and integration interfaces that drive the
formalized proofs.  The paper uses the real-derivative interface, which is
implemented via the more general Fr\'echet derivative.  The Archive of Formal
Proofs (AFP) entry for ordinary differential equations provides
Picard--Lindel\"of/flow infrastructure:
local flows, maximal existence intervals, and uniqueness for locally Lipschitz
vector fields.  Its SIR instantiation is introduced in
\Cref{sec:non-circular-globalization}.
Readers mainly interested in the SIR theorem package can skim the next two
subsections; their takeaway is that the checked interval theorems deliberately
use full-neighbourhood derivatives at endpoints, while FTC facts first arrive as
within-interval derivatives.

\subsection{Derivatives in Isabelle}
\label{sec:derivatives}

For real-valued functions $f : \R \to \R$, Isabelle provides the abbreviation
\[
  \isa{(f has\_real\_derivative D) (at x)}
\]
which asserts that $f$ has derivative $D$ at point $x$ (in the full
neighbourhood).  Such functions are total HOL functions on~$\R$; interval
statements restrict where the theorem concludes facts, not where the function
is defined.  Values outside the interval matter only insofar as they witness
these full-neighbourhood derivative assumptions.  The derivative abbreviation is
definitionally equal to
\isa{has\_field\_derivative}, which in turn reduces to the general
Fr\'echet derivative \isa{has\_derivative} with a linear multiplication map,
i.e. the linear map $h \mapsto D \cdot h$ (scalar multiplication by $D$).
In practice, this means that \isa{has\_real\_derivative} at an interior point
of $[a,b]$ uses a two-sided limit, not just a one-sided limit from within the
interval.  Thus, at the endpoints of $[a,b]$, the formal assumption is stronger
than a one-sided interval derivative: the function must be differentiable in a
full open neighbourhood of every point in $[a,b]$, including the endpoints.
This makes the scalar lemmas slightly stronger than the usual textbook interval
lemmas, but it is a deliberate interface choice rather than a mathematical
change to the SIR model.
In the SIR application below, trajectories arise from a flow defined on an open
existence interval, so the convention is satisfied on finite closed
subintervals after global existence is established.

Key combination rules used in our development, shown schematically:
\begin{itemize}
\item \isa{DERIV\_mult}: product rule,
  \[
    \frac{d}{dt}(f(t)g(t)) = f'(t)g(t) + f(t)g'(t).
  \]
\item \isa{DERIV\_chain2}: chain rule, $(g \circ f)' = g'(f(x)) \cdot f'(x)$.
\item \isa{DERIV\_isconst2}: the HOL-Analysis constancy rule used here; in the
  interval form relevant below, zero derivative on $(a,b)$ together with
  continuity on $[a,b]$ gives constancy on $[a,b]$.
  Schematically:
  \[
    \bigl(\forall x \in (a,b).\, f'(x)=0\bigr)
    \ \wedge\ f\text{ is continuous on }[a,b]
    \quad\Longrightarrow\quad
    \forall x,y\in[a,b].\, f(x)=f(y).
  \]
\item \isa{DERIV\_atLeastAtMost\_imp\_continuous\_on}: differentiability on
  $[a,b]$ implies continuity on $[a,b]$.
\end{itemize}

\subsection{Integration in Isabelle}
\label{sec:integration}

Our forward invariance proofs rely on the Fundamental Theorem of Calculus.
In HOL-Analysis, this is available as:
\begin{itemize}
\item \isa{integral\_has\_real\_derivative}: if $f$ is continuous on $[a,b]$
  and $t \in [a,b]$, then
  \[
    \lambda x.\, \int_a^x f(s)\,ds
  \]
  has derivative $f(t)$ at~$t$ within~$[a,b]$; this is a within-derivative,
  i.e. a derivative relative to the interval topology, not yet a
  full-neighbourhood derivative.
  Informally, the checked statement has the shape
  \[
    \bigl((\lambda x.\,\int_a^x f(s)\,ds)
      \ \text{\isa{has\_real\_derivative}}\ f(t)\bigr)
    \ \text{\isa{at}}\ t\ \text{\isa{within}}\ \{a..b\}.
  \]
  Isabelle's interval integral is the Henstock--Kurzweil gauge integral; we
  mention this only to identify the HOL-Analysis interface.  For the continuous
  functions used here on compact intervals, it agrees with the usual
  Riemann/Lebesgue value.
  The key interface subtlety for our proofs is the conversion between
  interval-relative and full-neighbourhood derivatives:
\item \isa{at\_within\_Icc\_at}: for $a < t < b$, the filter
  \isa{at t within \{a..b\}} equals the full \isa{at t}, enabling us to convert
  within-derivatives to full-neighbourhood derivatives in the interior.
  Mathematically, an interval-relative derivative and an unrestricted derivative
  coincide at interior points because a neighbourhood of such a point lies inside
  the interval.  This is a general topology/filter conversion, not an integration
  theorem, but it is paired here with the FTC fact because the conversion is used
  immediately after applying that fact.
\end{itemize}
In short, the scalar lemmas assume ordinary two-sided derivatives, the FTC first
supplies interval-relative derivatives, and our proofs convert those derivatives
only at interior points.

This conversion is one reason the framework records its endpoint derivative
convention explicitly in \Cref{sec:framework}.  The
product and chain rules in the integrating-factor proof are stated for full
derivatives, so the within-interval FTC derivative is converted only at interior
points.  The constancy argument uses full derivatives only on the open interval $(a,b)$;
endpoint differentiability is used mainly to obtain closed-interval continuity
in the artifact's chosen interface.

A within-derivative variant would assume less and therefore state a more general
theorem, but it was not needed for the AFP-flow application pursued here and
would require reworking the constancy and chain-rule steps against within
filters.

\subsection{Locales}
\label{sec:locales}

Isabelle's \emph{locale} mechanism~\cite{ballarin2014locales} provides
named proof contexts with fixed parameters and assumed properties; roughly,
locales are parameterized theories whose assumptions can later be instantiated.
Intuitively, a locale is a named bundle of assumptions; an interpretation proves
that the bundle holds in a concrete context and then makes all theorems proved
under those assumptions available there.
Theorems proved inside a locale are automatically parametric in those
parameters, enabling systematic reuse via locale interpretation or
direct rule application.
Locale interpretation discharges the assumptions of one locale in a new context
and transfers its derived lemmas; this is the mechanism behind the bridge from
scalar results to the global SIR flow in \Cref{sec:non-circular-globalization}.
For example, once the AFP SIR flow is shown to satisfy the assumptions of
\isa{SIR\_solution} on an interval, the theorems proved in that locale become
available for the flow on that interval.

In this paper, a generic scalar theorem proved in one locale becomes an SIR
theorem once a concrete SIR locale supplies the theorem's parameters and
assumptions.  The rest of the paper uses ``locale'' in this Isabelle sense.

\section{A Reusable Compartmental Framework}
\label{sec:framework}

One central contribution of this work is a focused, reusable library of generic
lemmas for compartment equations, derivative-sign arguments, and conservation.
The lemmas are used by the SIR instantiation and apply to equations of the form
$X'=f(t)X$ and to three-compartment conservation sums.
The Isabelle-facing packaging exposes these textbook arguments as reusable
locale-friendly interface lemmas with the derivative and continuity witnesses
needed by the SIR transfer proofs.

As outlined in \Cref{sec:introduction}, this framework fragment provides four
classes of results:
(1)~linear ODE solutions via the integrating factor,
(2)~scalar sign-preservation lemmas (composable to prove orthant invariance
for specific systems),
(3)~interface wrappers---convenience lemmas that package derivative-sign
conditions in the form used by the SIR proofs---for derivative-sign monotonicity,
and
(4)~conservation of compartment sums.

The present interface is deliberately scoped to SIR-relevant ingredients rather
than a general compartmental-model language.  The scalar sign and monotonicity
lemmas cover own-variable equations of the form $X'=f(t)X$ and derivative-sign
arguments, while the conservation theorem below has the three-compartment shape
needed for the checked SIR sign, conservation, and monotonicity obligations.
Within that scope, the scalar lemmas apply to later compartment proofs whose
ODE can isolate a state variable in the form $X'=f(t)X$ with continuous
coefficient $f$; for example, they cover the susceptible equation of
SI/SIR-style mass-action models and infectious equations such as those in
SIS/SIR models, once their coefficients are supplied by the surrounding model
proof.
The conservation component is also directly reusable for other
three-compartment conservation subgoals with the same zero-sum derivative shape,
even when additional sign or existence arguments remain model-specific.
Affine source terms and indexed compartment sums, such as SEIR or SIRS, require
the extensions listed in \Cref{sec:limitations}.

The framework theory uses the interfaces introduced in \Cref{sec:background}:
full-neighbourhood derivatives for the scalar locale assumptions,
within-interval FTC facts converted in the interior, and locales for
transferring generic sign and conservation lemmas to the SIR instance.

\subsection{Linear ODE Solution via Integrating Factor}

Many compartmental equations, including the $S$ and $I$ equations of basic SIR,
are \emph{homogeneous linear in their own variable}: $X'(t)=f(t)X(t)$ where
$f$ is a continuous coefficient function.  The classical integrating factor
method~\cite{teschl2012ode} yields the following integral representation.  This
representation is the exact equality form behind comparison bounds; here it
immediately gives sign preservation because the exponential factor is positive.
The formal statement uses the full-neighbourhood derivative convention of
\Cref{sec:derivatives}; one-sided endpoint variants are not formalized here.
This matches the checked Isabelle statement: the hypothesis supplies
derivatives on all of $[a,b]$, while the zero-derivative constancy step below
uses derivatives only on the open interval $(a,b)$.
The mathematical calculation is standard; the proof-engineering issue is
matching it to Isabelle's derivative filters and continuity obligations.

\begin{theorem}[\isa{linear\_ode\_solution}]
\label{thm:linear-ode}
Let $f, X : \R \to \R$, with $f$ continuous on $[a,b]$ and $a < b$, and
suppose that $X$ satisfies
\[
  (X\ \text{\isa{has\_real\_derivative}}\ f(t) \cdot X(t))
  \ \text{\isa{at}}\ t
\]
at each $t \in [a,b]$.
Define $F(t) = \int_a^t f(s)\, ds$.  The checked source uses Isabelle's
interval integral, which is oriented in general; the theorem conclusion here is
only for $t \in [a,b]$.  Consequently, for all $t \in [a,b]$:
\begin{equation}
\label{eq:linear-ode-representation}
  X(t) = X(a) \cdot \exp(F(t))
\end{equation}
\end{theorem}
The endpoint derivative assumption in this checked theorem is stronger than
one-sided endpoint differentiability; it follows the full-neighbourhood
interface stated in \Cref{sec:derivatives}.
We use the convention $F(t)=\int_a^t f(s)\,ds$, so the representation has the
positive exponent $X(a)\exp(F(t))$.
In Isabelle, $X$ is a total function $\R \to \R$; the displayed representation
is concluded only on $[a,b]$.
The equality is an explicit representation of any supplied solution, but it is
not a decoupled closed form because the coefficient $f$ may itself depend on the
trajectory.
In the SIR instantiation, the theorem is used with $X=S$ and $f=-\beta I$,
and with $X=I$ and $f=\beta S-\gamma$, to obtain the sign and integral facts
reported in \Cref{sec:sir-model}.

\begin{proof}
Define $g(t) = X(t) \cdot \exp(-F(t))$. Since the derivative hypothesis gives
differentiability at every point of $[a,b]$, $X$ is continuous there.  We show
$g'(t) = 0$ on $(a,b)$; continuity of \(f\) supplies the interval-integrability
facts behind \(F\):
\begin{enumerate}
\item The FTC (\isa{integral\_has\_real\_derivative}) first gives that $F$ has
  derivative $f(t)$ within $[a,b]$; at interior points, the interval-filter
  conversion changes that within-interval derivative to the full \isa{at t}
  derivative used by the product and chain rules.  Continuity of $F$ on
  $[a,b]$ is discharged by this within-interval derivative and
  \isa{DERIV\_continuous}.
  Equivalently, for every $t\in(a,b)$,
  \isa{at t within \{a..b\} = at t}, so $F'(t)=f(t)$ in the
  full-neighbourhood sense.
\item By the chain rule, $\frac{d}{dt}[\exp(-F(t))] = -f(t) \cdot \exp(-F(t))$.
\item For each $t\in(a,b)$, the product rule and the ODE hypothesis
  $X'(t)=f(t)X(t)$ give:
\[
  g'(t) = f(t) X(t) \cdot \exp(-F(t)) + X(t) \cdot (-f(t)) \exp(-F(t)) = 0
\]
\end{enumerate}
Since $X$ is continuous on $[a,b]$ by the derivative hypothesis and
\isa{DERIV\_\allowbreak atLeastAtMost\_\allowbreak imp\_\allowbreak continuous\_on},
$g$ is continuous there as a product of continuous functions.  With zero
derivative on $(a,b)$, \isa{DERIV\_isconst2} gives $g(t) = g(a)$.
Now $g(a) = X(a) \cdot \exp(-F(a)) = X(a) \cdot \exp(0) = X(a)$ (since
$F(a) = \int_a^a f = 0$). Therefore $X(t) \cdot \exp(-F(t)) = X(a)$,
and multiplying both sides by $\exp(F(t))$ (which is nonzero) gives
$X(t) = X(a) \cdot \exp(F(t))$.
\end{proof}

\medskip
The paper calculation is short; the checked lemma's work is to package total
functions, closed-interval hypotheses, endpoint derivative conventions, and
explicit derivative and continuity witnesses into one locale-friendly rule, so
later SIR proofs can instantiate the calculation without replaying
automation-sensitive calculus search.
The key technical points in the Isabelle proof:
\begin{itemize}
\item The FTC derivative is first a derivative \emph{within} $[a,b]$; the proof
  uses the full-neighbourhood conversion only in the interior.
\item The chain rule for $\exp(-F(t))$ is applied via \isa{DERIV\_chain2}
  with explicit instantiation, avoiding the timeout-prone
  \isa{auto intro!:\ derivative\_intros}.  In the artifact this occurs inside
  \isa{linear\_ode\_solution} rather than being hidden behind automation.
\item Continuity of $\exp(-F)$ uses \isa{continuous\_on\_exp} applied to
  the negated integral function directly.
\end{itemize}

\subsection{Scalar Sign Preservation}

The integral representation immediately yields sign preservation:

\begin{corollary}[\isa{linear\_ode\_nonneg}]
\label{cor:nonneg}
Under the hypotheses of \Cref{thm:linear-ode}, if $X(a) \ge 0$ then
$X(t) \ge 0$ for all $t \in [a,b]$.
\end{corollary}

\begin{proof}
By \Cref{eq:linear-ode-representation},
$X(t) = X(a) \cdot \exp(F(t))$. Since $\exp(F(t)) > 0$ and
$X(a) \ge 0$, the product is nonnegative.
\end{proof}

\begin{corollary}[\isa{linear\_ode\_pos}]
\label{cor:pos}
Under the same hypotheses, if $X(a) > 0$ then $X(t) > 0$ for all $t \in [a,b]$.
\end{corollary}
\begin{proof}
The same representation gives $X(t)=X(a)\exp(F(t))$, a product of two positive
terms.
\end{proof}

Theorem~\ref{thm:linear-ode} and these corollaries are designed to apply to any
differentiable solution of $X'=f(t)X$ with continuous $f$, independently of how
that solution is supplied.  The checked reuse demonstrated here is the SIR
instantiation; a different model, such as SI, would still require its own
mechanized existence/flow bridge before making a checked model claim.

\subsection{Nonnegativity from Nonnegative Derivative}

For compartments like $R$ in the SIR model, whose derivatives are nonnegative
(not linear in the compartment itself), we provide:

\begin{lemma}[\isa{nonneg\_deriv\_nonneg}]
\label{lem:nonneg-deriv}
Let $X,g : \R \to \R$ and $a < b$.  Suppose that, for every
$t\in[a,b]$, $X$ has full-neighbourhood derivative $g(t)$ at $t$, using
the same \isa{has\_real\_derivative} convention as \Cref{thm:linear-ode}.
If $g(t) \ge 0$ for all $t \in [a,b]$ and $X(a) \ge 0$, then
$X(t) \ge 0$ for all $t \in [a,b]$.
\end{lemma}

\begin{proof}
Since $g(t)\ge0$ on $[a,b]$, the Mean-Value-Theorem consequence
\isa{nondecreasing\_\allowbreak from\_\allowbreak nonneg\_\allowbreak derivative}
gives that $X$ is nondecreasing.
Hence $X(t) \ge X(a) \ge 0$.
\end{proof}
The assumption $a<b$ matches the checked nondegenerate-interval theorem used by
the monotonicity wrapper; at the endpoint $t=a$, the conclusion is just the
initial premise.
Thus \isa{nonneg\_deriv\_nonneg} is a monotonicity-implies-nonnegativity wrapper,
not a Nagumo-style boundary-invariance argument.

\subsection{Three-Compartment Conservation}

\begin{theorem}[\isa{three\_compartment\_conservation}]
\label{thm:conservation}
Let $f, g, h, D_f, D_g, D_h : \R \to \R$ and let $a,b\in\R$ with $a < b$.
Suppose that, for every $s\in[a,b]$, the functions $f$, $g$, and $h$ have
full-neighbourhood derivative values $D_f(s)$, $D_g(s)$, and $D_h(s)$,
respectively, using the same derivative convention as \Cref{thm:linear-ode}.
If $D_f(s) + D_g(s) + D_h(s) = 0$ for every $s\in[a,b]$, then
\[
  f(t) + g(t) + h(t) = f(a) + g(a) + h(a) \quad \text{for all } t \in [a,b].
\]
\end{theorem}
We specialize to three compartments here because this is the shape needed by
SIR; \Cref{sec:limitations} records the indexed $n$-compartment generalization
as future work.
The separate theorem exposes the derivative witnesses and zero-sum obligation in
a reusable Isabelle interface.
The assumption $a<b$ matches the HOL-Analysis constancy lemma; at $a=b$, the
conclusion is the trivial identity
$f(a)+g(a)+h(a)=f(a)+g(a)+h(a)$.

\begin{proof}
Define $N(s) = f(s) + g(s) + h(s)$. By the sum rule,
$N'(s) = D_f(s) + D_g(s) + D_h(s) = 0$. By continuity (from differentiability)
and \isa{DERIV\_isconst2}: $N(t) = N(a)$.
\end{proof}

\subsection{Monotonicity Wrappers}

\begin{corollary}[\isa{nonincreasing\_from\_nonpos\_derivative}]
Let $f,D:\R\to\R$. If $s \le t$, and for every $u\in[s,t]$ the function
$f$ has full-neighbourhood derivative $D(u)$ at $u$ with $D(u) \le 0$,
then $f(t) \le f(s)$.
\end{corollary}

\begin{corollary}[\isa{nondecreasing\_from\_nonneg\_derivative}]
Let $f,D:\R\to\R$. If $s \le t$, and for every $u\in[s,t]$ the function
$f$ has full-neighbourhood derivative $D(u)$ at $u$ with $D(u) \ge 0$,
then $f(s) \le f(t)$.
\end{corollary}

The displayed corollaries wrap the corresponding HOL-Analysis monotonicity
lemmas, which are direct Mean Value Theorem consequences.  In the checked file,
they are short corollaries of \isa{DERIV\_nonpos\_imp\_nonincreasing} and
\isa{DERIV\_nonneg\_imp\_nondecreasing}.  They expose the
derivative witness, sign premise, continuity facts, and interval-order premises
explicitly, turning each SIR monotonicity proof into the local obligations
``provide the derivative expression'' and ``prove its sign.''
They are used for $S$ and $R$ monotonicity and for interval-wide monotonicity of
$I$ under $\Rinit \le 1$.
Here ``has derivative'' uses the same full-neighbourhood
\isa{has\_real\_derivative} convention as \Cref{thm:linear-ode}; the interval
$[s,t]$ scopes the sign premise and conclusion.
When $s=t$, the conclusions are reflexive; the nondegenerate case is the usual
derivative-sign monotonicity argument.
They are reusable within developments that adopt this artifact's
derivative-interface convention.  A more general HOL-Analysis or AFP version
would be natural future library work, but would require a broader API decision.

\subsection{Design Decisions}

\paragraph{Why keep the integrating-factor proof?}
The SIR existence theory now uses the AFP Picard--Lindel\"of/flow
infrastructure~\cite{immler2012ode} to construct the unique solution.  We
nevertheless keep the sign-preservation proof separate because it exploits the
\emph{homogeneous linear} structure of the relevant compartmental equations and
because the framework theory imports no SIR-specific material.  The broader
design trade-off is discussed in \Cref{sec:discussion-if-vs-pl}.

\section{SIR Model Instantiation}
\label{sec:sir-model}

\subsection{The SIR ODE System}
\label{sec:sir-ode-system}

The mass-action classical SIR model~\cite{kermack1927sir} partitions a closed
population with no births, deaths, immigration, or emigration (i.e., no vital
dynamics) into susceptible~($S$), infectious~($I$), and recovered~($R$)
compartments (\Cref{fig:sir-flow}):
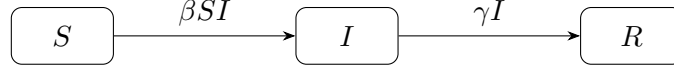
\begin{figure}[ht]
\centering
\begin{tikzpicture}[
  >=Stealth,
  node distance=2.4cm,
  compartment/.style={draw, rounded corners, minimum width=1.35cm, minimum height=0.75cm}
]
\node[compartment] (S) {$S$};
\node[compartment, right=of S] (I) {$I$};
\node[compartment, right=of I] (R) {$R$};
\draw[->] (S) -- node[above] {$\beta S I$} (I);
\draw[->] (I) -- node[above] {$\gamma I$} (R);
\end{tikzpicture}
\caption{Closed-population SIR compartment flow.  The transmission term moves
individuals from susceptible to infectious, and the recovery term moves
individuals from infectious to recovered.}
\label{fig:sir-flow}
\end{figure}
\begin{align}
  \deriv{S}{t} &= -\beta S I \label{eq:sir-s} \\
  \deriv{I}{t} &= \beta S I - \gamma I \label{eq:sir-i} \\
  \deriv{R}{t} &= \gamma I \label{eq:sir-r}
\end{align}
We suppress the explicit time argument in these displayed equations:
$S$, $I$, and $R$ denote $S(t)$, $I(t)$, and $R(t)$, while the derivative
notation denotes the derivative at that same time.
For the forward initial-value problem used by \isa{SIR\_ODE}, these equations
are paired with
\[
  S(0)=S_0,\qquad I(0)=I_0,\qquad R(0)=\Rrecinit,
  \qquad S_0,I_0,\Rrecinit\ge0.
\]
Here \(\Rrecinit\) is the recovered initial value.  For this forward problem,
$N_0=S_0+I_0+\Rrecinit$ denotes the total initial population; finite-interval
statements below write
$N=S(a)+I(a)+R(a)$.
When the flow bridge specializes a finite interval to $a=0$, these agree:
$N=N_0$.
The \isa{SIR\_solution} locale uses a generic finite interval $[a,b]$, enabling
reuse on arbitrary subintervals of a trajectory; the \isa{SIR\_ODE} flow results
specialize it to the forward initial time $0$.
Thus $[a,b]$ denotes the generic finite-interval locale below, while $[0,b]$
denotes the forward-flow intervals obtained from \isa{SIR\_ODE}.
The main notation used below is:
\begin{center}
\small
\begin{tabular}{@{}lp{0.68\linewidth}@{}}
\multicolumn{2}{@{}l}{\emph{Forward problem}} \\
$R(t)$ & recovered compartment at time $t$ \\
$\Rrecinit$ & recovered initial value in the forward state triple \\
$N_0$ & $S_0+I_0+\Rrecinit$, the forward problem's initial total population \\
\addlinespace
\multicolumn{2}{@{}l}{\emph{Finite interval}} \\
$N$ & $S(a)+I(a)+R(a)$ on a finite interval $[a,b]$ \\
\addlinespace
\multicolumn{2}{@{}l}{\emph{Threshold ratios}} \\
$\Rinit$ & $\beta S(a)/\gamma$, initial effective threshold ratio \\
$\Reff(t)$ & $\beta S(t)/\gamma$, time-dependent threshold ratio \\
$\Rclassical$ & $\beta N/\gamma$, fully susceptible disease-free
next-generation quantity for the same total population $N$ at the DFE
$(S,I,R)=(N,0,0)$; for the forward flow, $N=N_0$; see \Cref{sec:R0}
\end{tabular}
\end{center}
The notation $\Rrecinit$ is always the recovered initial value, never a
reproduction number; the paper avoids the symbol $R_0$ for both recovered initial
values and threshold quantities because of its epidemiological association with
reproduction numbers.
Throughout, italic \(R\) denotes the recovered compartment, while calligraphic
\(\mathcal R\) denotes threshold ratios.
The parameters satisfy $\beta > 0$ and $\gamma > 0$: $\beta$ is the mass-action
transmission coefficient, $\gamma$ is the recovery rate, and the infection term
$\beta SI$ is proportional to the product of susceptible and infectious
compartment sizes.  Strict positivity models active transmission and active
recovery; degenerate boundary cases such as $\beta=0$ or $\gamma=0$ are outside
the artifact's stated scope.  The compartments are real-valued state variables
representing counts on an unnormalized scale; with count-valued compartments,
$\gamma$ has reciprocal-time units and $\beta$ has reciprocal population-time
units.  If the compartments are normalized to fractions, the population size is
absorbed into the transmission parameter.
Standard-incidence formulations such as $\beta SI/N$ are outside the present
artifact.  Isabelle types carry no physical units, so epidemiological units are
interpretive metadata; \Cref{sec:R0} uses the usual mass-action units only when
discussing dimensionless threshold quantities.

\subsection{The Locales \isa{SIR\_solution} and \isa{SIR\_ODE}}

We formalize solutions as a locale fixing $\beta$, $\gamma$, $S$, $I$, $R$,
$a$, $b$ and assuming:
\begin{itemize}
\item Positivity of parameters: $\beta > 0$, $\gamma > 0$.
\item A non-trivial interval: $a < b$.
\item The ODE system holds pointwise on $[a,b]$ (in full-neighbourhood
  \isa{has\_real\_derivative} form; see
  \Cref{sec:derivatives}), so the formal trajectory is treated as a
  real-valued function defined beyond the closed interval because Isabelle's
  \isa{at t} filter quantifies over punctured neighbourhoods in~$\R$.
\item Nonnegative \emph{initial conditions}: $S(a) \ge 0$, $I(a) \ge 0$, $R(a) \ge 0$.
\end{itemize}

Crucially, \isa{SIR\_solution} does \emph{not} assume nonnegativity of the full
trajectory.  Instead, we derive it from the initial conditions using the
framework's scalar sign-preservation lemmas (\Cref{cor:nonneg}).
The locale header is shown below as an ASCII-transliterated schematic of the
checked locale header, not as literal Isabelle source; \isa{!!} denotes
meta-level universal quantification, \isa{:} denotes membership, and \isa{==>}
denotes meta-level implication.  Only Isabelle symbols and line breaks are
simplified; the assumptions are unchanged.  The exact source header is in
\path{theories/SIR/SIR_Defs.thy}.
\begin{quote}
\small
\begin{verbatim}
locale SIR_solution =
  fixes beta gamma :: real
    and S I R :: "real => real"
    and a b :: real
  assumes pos_beta: "0 < beta"
    and pos_gamma: "0 < gamma"
    and interval: "a < b"
    and ode_S: "!!t. t : {a..b} ==>
      (S has_real_derivative (- beta * S t * I t)) (at t)"
    and ode_I: "!!t. t : {a..b} ==>
      (I has_real_derivative (beta * S t * I t - gamma * I t)) (at t)"
    and ode_R: "!!t. t : {a..b} ==>
      (R has_real_derivative (gamma * I t)) (at t)"
    and init_S: "0 <= S a"
    and init_I: "0 <= I a"
    and init_R: "0 <= R a"
\end{verbatim}
\end{quote}
Mathematically, \isa{SIR\_solution} is a strengthened Isabelle interface for a
classical SIR solution on the closed interval $[a,b]$: it records positive
parameters, the three differential equations in Isabelle's full-neighbourhood
derivative interface, and nonnegative initial data, but not nonnegativity of
later states.

The companion locale \isa{SIR\_ODE} fixes only $\beta$, $\gamma$, and the
initial values $S_0,I_0,\Rrecinit$ with $\beta,\gamma>0$ and nonnegative initial
data.

The construction lives in theory \isa{SIR\_Existence}, which supplies
the \isa{SIR\_ODE} locale used here: it defines the SIR vector field on
$\R^3$, proves it continuously differentiable, instantiates the AFP ODE flow
infrastructure, and defines scalar components of the unique local flow.
This infrastructure supplies a maximal existence interval and a unique local
flow for locally Lipschitz vector fields.
Uniqueness comes from the same AFP Picard--Lindel\"of/flow infrastructure:
the polynomial vector field is continuously differentiable, hence locally
Lipschitz on the finite-dimensional state space.

When a checked identifier contains the historical name \isa{R\_zero}, the paper
renders that quantity mathematically as $\Rinit$ unless the Isabelle name itself
is being cited.  The source identifier is retained for artifact traceability; its
epidemiological interpretation follows the prose notation rather than that
historical source identifier.

\subsection{Non-circular Globalization and Flow Bridge}
\label{sec:non-circular-globalization}

Global forward existence means that all nonnegative times belong to the AFP
flow's existence interval for the SIR initial state.  We prove it from
conservation and nonnegativity before any global qualitative theorem is
available.  \Cref{fig:proof-order} summarizes the proof dependencies.
For orientation, the proof obligations are: construct the AFP local flow; prove
sign and conservation for each already-existing local-flow prefix; use those
prefix facts to place the orbit in one compact simplex; invoke compact
continuation to get all forward times; and then package the scalar components as
a \isa{SIR\_solution} instance on every nontrivial finite forward interval
$[0,b]$ with $b>0$.

\paragraph{Strategy.}
The AFP infrastructure first gives a unique local flow for the polynomial vector
field, together with a maximal existence interval.  On each finite prefix that is
already known to lie in this interval, the flow-component ODE lemmas let us reuse
the scalar sign-preservation and conservation arguments.  They show that the
prefix remains in the compact simplex
\[
  \{(S,I,R)\in\R^3 \mid S\ge 0,\ I\ge 0,\ R\ge 0,\ S+I+R=S_0+I_0+\Rrecinit\}.
\]
Conservation and nonnegativity keep the trajectory inside this compact simplex,
ruling out finite-time escape.  This proof order avoids assuming global
nonnegativity in the proof of global existence.  The escape alternative is that
a finite right endpoint would force the trajectory eventually to leave every
compact subset of the domain.  The simplex is closed and bounded in
finite-dimensional Euclidean space, hence compact.  Geometrically, it
is the scaled 2-simplex in $\R^3$ obtained by intersecting the nonnegative
orthant with the hyperplane $S+I+R=S_0+I_0+\Rrecinit$.
In particular, with $N_0=S_0+I_0+\Rrecinit$, each compartment lies in $[0,N_0]$, so
the simplex is contained in the box $[0,N_0]^3$.

\paragraph{Existing-prefix invariant.}
For any existing forward time $t$, membership of $t$ in the maximal existence
set gives membership of every time in the prefix $[0,t]$, because that set is an
interval containing the initial time $0$.  Applying the scalar lemmas on that
prefix gives containment at $t$.
The case $t=0$ is immediate from the initial state; scalar-locale applications
use nondegenerate prefixes when $t>0$.
Concretely, \isa{sir\_flow\_\allowbreak component\_\allowbreak deriv} supplies the component ODEs for
the local AFP flow.  For an arbitrary forward time already known to lie in the
maximal interval, we restrict the flow to the finite interval ending at that
time.  The component nonnegativity lemmas
\isa{sir\_flow\_\allowbreak nonneg\_S}, \isa{sir\_flow\_\allowbreak nonneg\_I}, and
\isa{sir\_flow\_\allowbreak nonneg\_R}, together with \isa{sir\_flow\_\allowbreak conservation},
then establish \isa{sir\_flow\_\allowbreak in\_simplex} for that existing flow
point.

\paragraph{Compact continuation.}
Only after proving this existing-prefix invariant do we instantiate the AFP
theorem \isa{flow\_in\_\allowbreak compact\_\allowbreak right\_\allowbreak existence}.  This
compact-continuation principle says that a local flow whose forward orbit, at
every nonnegative time already in its maximal existence interval, remains in one
compact subset of the state space has no finite right endpoint:
\begin{center}
\small\isa{flow\_in\_compact\_right\_existence} $\Rightarrow$
\isa{sir\_global\_existence}.
\end{center}
We use the \isa{flow\_\allowbreak in\_\allowbreak compact\_\allowbreak right\_\allowbreak existence}
theorem from the AFP~2024 ODE formalization~\cite{afp-ode}, matching standard
ODE presentations such as Teschl's~\cite{teschl2012ode}.  In this SIR
instantiation, the vector field is polynomial on all of $\R^3$, so the compact
simplex automatically lies inside the ambient domain of the local-flow theorem.
Thus there is no boundary-of-domain obstruction in this application; the
continuation argument only has to rule out finite escape from compact sets.
This is the standard maximal-solution continuation criterion: containment of an
existing forward orbit in one compact subset of the open domain excludes a finite
right endpoint.
Equivalently, a finite right endpoint would force the existing orbit to leave
every compact subset of the open domain as the endpoint is approached.
The right-maximal interval and local-flow premises are supplied by the AFP flow
locale; the SIR proof only supplies the compact witness and containment premises
listed below.
The SIR-specific work is not the continuation principle itself, but constructing
the compact simplex witness from sign and conservation facts already proved for
existing local-flow prefixes.  In the SIR instantiation, the ambient domain is
the open set $U=\R^3$ for the polynomial vector field, $x(t)$ is the AFP local
flow from the initial state, $J$ is its maximal forward existence interval with
$0\in J$, and the continuation fact has the following paper-level shape:
\[
\Bigl(K\subseteq U\text{ is compact}\;\wedge\;
  \forall t\ge0.\;(t\in J \Rightarrow x(t)\in K)\Bigr)
\quad\Longrightarrow\quad
[0,\infty)\subseteq J .
\]
This display is a prose specialization of
\isa{flow\_in\_\allowbreak compact\_\allowbreak right\_\allowbreak existence}, not a
verbatim transcription of the AFP theorem.
Here $J$ corresponds to \isa{sir\_c1.existence\_ivl0}, the AFP existence interval
for the local flow from the initial SIR state.
Thus the premise $t\in J$ is exactly the condition that $t$ is an
already-existing forward time of the local flow.
As used here, $J$ is the AFP maximal existence interval for that local flow; the
AFP flow locale supplies its interval, openness, and maximality facts, and the
SIR argument uses the forward part from the initial time.
In the checked invocation, the AFP flow locale supplies the local-flow,
uniqueness, maximal-interval, openness, and continuity premises; the SIR proof
discharges the compactness of the simplex, containment of every existing forward
flow point in that simplex, the ambient-domain inclusions, and the requested
$t\ge0$.
The premise discharge can be read from the following local map:
\begin{center}
\small
\begin{tabular}{@{}>{\raggedright\arraybackslash}p{0.36\linewidth}
                >{\raggedright\arraybackslash}p{0.52\linewidth}@{}}
\toprule
\textbf{Continuation ingredient} & \textbf{SIR discharge} \\
\midrule
Open ambient domain $U$ & $U=\R^3$ for the polynomial SIR vector field. \\
Local flow and maximal interval $J$ &
AFP flow locale \isa{sir\_c1} and interval \isa{sir\_c1.existence\_ivl0}. \\
Compact witness $K\subseteq U$ &
The simplex \isa{sir\_simplex (x0\$1 + x0\$2 + x0\$3)} and
\isa{sir\_simplex\_compact}. \\
Forward-orbit containment in $K$ &
\isa{sir\_flow\_in\_simplex}, proved from local-flow sign facts and conservation. \\
All-forward-times conclusion &
\isa{sir\_global\_existence}. \\
\bottomrule
\end{tabular}
\end{center}
In ordinary mathematical terms, the invocation applies the continuation
criterion to the already-built local flow using the conserved nonnegative
simplex as the compact set containing every existing forward prefix.
That simplex depends only on the initial total population, not on the chosen
prefix endpoint.
The unabridged checked proof is the theorem \isa{sir\_global\_existence} in
\path{work/SIR_Existence.thy}.  Its central Isabelle invocation is short; below
is a complete, lightly transliterated excerpt, with Isabelle's symbolic escapes
rendered as ASCII.  In Isabelle's vector syntax, \isatt{x0\$1}, \isatt{x0\$2},
and \isatt{x0\$3} are the first, second, and third components of the initial
vector \isatt{x0}; \isatt{in} denotes membership.  Here \isa{sir\_c1} names the
AFP flow-locale instance for the SIR vector field:
\begin{quote}
\footnotesize
\begin{verbatim}
theorem sir_global_existence:
  assumes "0 < beta" "0 < gamma"
    and nonneg: "0 <= x0$1" "0 <= x0$2" "0 <= x0$3"
    and "0 <= t"
  shows "t in sir_c1.existence_ivl0 beta gamma x0"
proof (rule sir_c1.flow_in_compact_right_existence
       [of beta gamma x0 "sir_simplex (x0$1 + x0$2 + x0$3)"])
  fix s assume "0 <= s" "s in sir_c1.existence_ivl0 beta gamma x0"
  then show "sir_c1.flow0 beta gamma x0 s
    in sir_simplex (x0$1 + x0$2 + x0$3)"
    using sir_flow_in_simplex[OF assms(1-5)] by auto
next
  show "compact (sir_simplex (x0$1 + x0$2 + x0$3))"
    using nonneg by (intro sir_simplex_compact) linarith
next
  show "sir_simplex (x0$1 + x0$2 + x0$3) <= UNIV" by simp
next
  show "x0 in (UNIV :: (real^3) set)" by simp
next
  show "0 <= t" by fact
qed
\end{verbatim}
\end{quote}
The displayed obligations match the preceding table: compact containment for
each existing forward time, compactness of the simplex, routine domain
inclusions, and the requested nonnegative target time.  The SIR-specific witness
is the compact simplex and the already-proved containment of every existing
forward flow point in it.  In the excerpt, \isa{assms(1-5)} are precisely the
parameter-positivity and initial-nonnegativity assumptions, excluding the target
time assumption \isa{0 <= t}.
Concretely, \isa{sir\_flow\_\allowbreak in\_\allowbreak simplex} states that any existing
forward point of \isa{sir\_c1.flow0} lies in
\isa{sir\_simplex (x0\$1 + x0\$2 + x0\$3)}, and its proof combines the three
flow nonnegativity lemmas with \isa{sir\_flow\_\allowbreak conservation}; the compact
witness itself is discharged by \isa{sir\_simplex\_\allowbreak compact}.
This is the formal theorem corresponding to the escape-alternative principle
used in the strategy paragraph.

\paragraph{Flow bridge.}
The simplex is the single compact witness containing every existing forward flow
point.  \isa{SIR\_ODE.\allowbreak global\_\allowbreak existence}
records the resulting all-forward-times theorem in the SIR locale.  After that,
every interval $[0,b]$ with $b>0$ lies in the existence interval, and
\isa{SIR\_ODE.\allowbreak sir\_\allowbreak solution\_on\_\allowbreak interval} packages the scalar
components as an instance of \isa{SIR\_solution} on $[0,b]$.
Here ``flow bridge'' is paper shorthand for this packaging theorem and proof
pattern, not a separate Isabelle locale or definition.

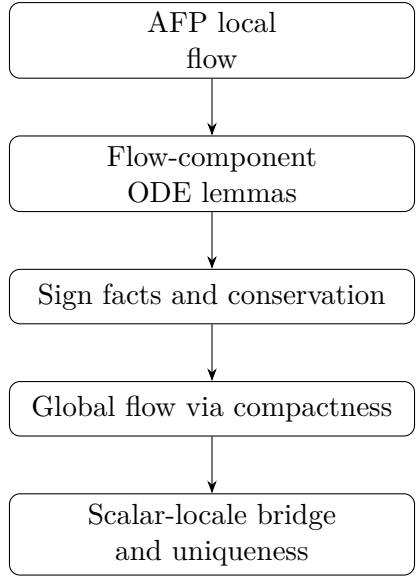
\begin{figure}[ht]
\centering
\begin{tikzpicture}[
  >=Stealth,
  node distance=0.75cm,
  proofstep/.style={draw, rounded corners, align=center, text width=5.1cm, minimum height=0.72cm}
]
\node[proofstep] (local) {AFP local\\flow};
\node[proofstep, below=of local] (components) {Flow-component ODE lemmas};
\node[proofstep, below=of components] (qual) {Sign facts and conservation};
\node[proofstep, below=of qual] (global) {Global flow via compactness};
\node[proofstep, below=of global] (scalar) {Scalar-locale bridge and uniqueness};
\draw[->] (local) -- (components);
\draw[->] (components) -- (qual);
\draw[->] (qual) -- (global);
\draw[->] (global) -- (scalar);
\end{tikzpicture}
\caption{Proof order for avoiding circularity.  Arrows denote proof dependency:
sign facts are first proved for the local flow, then used with conservation to
justify the global flow; the scalar-locale bridge comes after global existence.}
\label{fig:proof-order}
\end{figure}
\FloatBarrier

\begin{theorem}[Flow bridge and uniqueness]
\label{thm:flow-bridge}
Assume $\beta,\gamma>0$ and $S_0,I_0,\Rrecinit\ge 0$.  For every $b>0$ (only
because \isa{SIR\_solution} requires a nontrivial interval), the scalar
components of the AFP flow generated by the SIR vector field from
$(S_0,I_0,\Rrecinit)$ instantiate \isa{SIR\_solution} on $[0,b]$:
\[
  \text{\isa{SIR\_solution}}
  (\beta,\gamma,S_{\mathrm{flow}},I_{\mathrm{flow}},R_{\mathrm{flow}},0,b).
\]

\medskip
\noindent\emph{Uniqueness.}  Let
$S^\ast,I^\ast,R^\ast:\R\to\R$ be total functions satisfying
\Cref{eq:sir-s,eq:sir-i,eq:sir-r} on $[0,b]$ in the same full-neighbourhood
derivative convention and with the same initial conditions
$S^\ast(0)=S_0$, $I^\ast(0)=I_0$, and
$R^\ast(0)=\Rrecinit$.  Then this trajectory agrees pointwise with those scalar
components on $[0,b]$.  This is not a uniqueness claim for weaker
Carath\'eodory or endpoint-only solution notions.
\end{theorem}
In the checked uniqueness proof, the compared scalar triple is converted to the
vector trajectory $t\mapsto(S^\ast(t),I^\ast(t),R^\ast(t))$; the
full-neighbourhood derivative assumptions on $[0,b]$ supply the vector-ODE
premises required by AFP local uniqueness.
Concretely, the proof defines a vector function with these three scalar
components, converts each scalar \isa{has\_real\_derivative} assumption to a
one-dimensional \isa{has\_derivative} fact on the interval, and combines the
three component derivatives to obtain the required
\isa{has\_vector\_derivative} statement for the SIR vector field.  The equal
initial conditions give the same initial vector state as the AFP flow.
This uniqueness statement is intentionally scoped to the same regularity
convention as \isa{SIR\_solution}.

\begin{proof}[Proof sketch]
The preceding construction proves the theorem in the checked artifact.  AFP
first supplies the local flow and its component ODE lemmas.  On each already
existing forward prefix, the scalar nonnegativity and conservation facts put the
flow point in the compact simplex displayed above.  The AFP compact-continuation
theorem then gives all forward times, and the scalar components of that global
flow instantiate \isa{SIR\_solution} on every nontrivial interval $[0,b]$.
The premise discharge is: the polynomial vector field and AFP flow locale give
the local-flow, continuity, openness, and maximal-interval premises; the simplex
is the compact set; \isa{sir\_flow\_in\_simplex} gives compact containment for
every existing forward time; and global existence supplies the interval
membership needed to interpret \isa{SIR\_solution}.
Uniqueness follows by packaging any compared scalar triple as a vector-valued
trajectory and applying the AFP Picard--Lindel\"of uniqueness theorem for the
locally Lipschitz SIR vector field with the same initial state and component
ODEs.
\end{proof}

This corresponds to \isa{SIR\_ODE.sir\_solution\_on\_interval} and
\isa{SIR\_ODE.sir\_uniqueness}.  The case $b=0$ is excluded only because
\isa{SIR\_solution} is a nontrivial-interval locale.

\subsection{Formalization Architecture}

The Isabelle artifact separates generic calculus facts, SIR-specific
qualitative reasoning, and the AFP flow construction.  \Cref{sec:artifact}
records the exact commit, build command, validation environment, and artifact
statistics.

\begin{table}[ht]
\centering
\small
\begin{tabular}{>{\raggedright\arraybackslash}p{0.34\linewidth}
                >{\raggedright\arraybackslash}p{0.56\linewidth}}
\toprule
\textbf{Theory area} & \textbf{Role} \\
\midrule
Framework &
\path{theories/Framework/Compartmental_Model.thy}: integrating-factor,
sign-preservation, monotonicity, and three-compartment conservation lemmas. \\
SIR qualitative theories &
\path{theories/SIR/SIR_*.thy}: scalar SIR locale, forward invariance, conservation,
monotonicity, phase-plane invariant, threshold, stationary-infection, and
boundedness theorems. \\
Existence bridge &
\path{work/SIR_Existence.thy}: checked session source for the polynomial vector field,
Picard--Lindel\"of/flow instantiation, global forward existence,
scalar-flow bridge, and uniqueness theorem. \\
\bottomrule
\end{tabular}
\caption{Artifact organization.  The SIR theories reuse the generic framework;
the existence theory connects those qualitative theorems to the unique AFP
flow.}
\label{tab:artifact-architecture}
\end{table}

Readers interested only in the mathematical content may read the displayed
theorem statements and skip the Isabelle-name traceability table.
The displayed theorem statements below are prose-facing summaries.  When they
omit Isabelle syntax or locale packaging, they do not intentionally omit
mathematical side conditions; the exact checked statements live in the following
Isabelle constants and source files:
\begin{table}[ht]
\centering
\footnotesize
\begin{tabular}{>{\raggedright\arraybackslash}p{0.21\linewidth}
                >{\raggedright\arraybackslash}p{0.38\linewidth}
                >{\raggedright\arraybackslash}p{0.31\linewidth}}
\toprule
\textbf{Paper result} & \textbf{Isabelle names} & \textbf{Source file(s)} \\
\midrule
Integrating factor and sign preservation &
\isa{linear\_ode\_\allowbreak solution}, \isa{linear\_ode\_\allowbreak nonneg},
\isa{linear\_ode\_\allowbreak pos} &
\path{theories/Framework/}\newline\path{Compartmental_Model.thy} \\
Derivative sign and conservation framework &
\isa{nonneg\_deriv\_\allowbreak nonneg},
\isa{nonincreasing\_from\_\allowbreak nonpos\_derivative},
\isa{three\_compartment\_\allowbreak conservation} &
\path{theories/Framework/}\newline\path{Compartmental_Model.thy} \\
Global SIR flow and bridge &
\isa{sir\_global\_\allowbreak existence},
\isa{SIR\_ODE.sir\_\allowbreak solution\_on\_\allowbreak interval},
\isa{SIR\_ODE.sir\_\allowbreak uniqueness} &
\path{work/SIR_Existence.thy} \\
Forward invariance and integral representations &
\isa{SIR\_solution.I\_\allowbreak nonneg}, \isa{SIR\_solution.S\_\allowbreak nonneg},
\isa{SIR\_solution.R\_\allowbreak nonneg},
\isa{SIR\_solution.I\_\allowbreak solution}, \isa{SIR\_solution.S\_\allowbreak solution} &
\path{theories/SIR/}\newline\path{SIR_Forward_Invariance.thy} \\
Conservation, monotonicity, and bounds &
\isa{SIR\_solution.conservation},
\isa{SIR\_solution.S\_\allowbreak nonincreasing},
\isa{SIR\_solution.R\_\allowbreak nondecreasing},
\isa{SIR\_solution.compartment\_\allowbreak bounds} &
\path{theories/SIR/}\newline
\path{SIR_Conservation.thy},
\path{SIR_Monotonicity.thy},
\path{SIR_Invariant.thy} \\
Phase-plane invariant and threshold facts &
\isa{SIR\_solution.KM\_\allowbreak invariant\_value},
\isa{SIR\_solution.epidemic\_\allowbreak growth\_R\_eff},
\isa{SIR\_solution.I\_\allowbreak nonincreasing\_if\_\allowbreak R\_zero\_le\_one},
\isa{SIR\_solution.peak\_\allowbreak iff} &
\path{theories/SIR/}\newline
\path{SIR_Phase_Plane.thy},
\path{SIR_Threshold.thy},
\path{SIR_Peak.thy} \\
\bottomrule
\end{tabular}
\caption{Traceability from paper-level result families to checked Isabelle names
and source files.  Paths are relative to the artifact repository root.
Locale-qualified names are theorems in \isa{SIR\_solution}; the flow row is the
\isa{SIR\_ODE} transfer layer.  Names containing \isa{R\_zero} are source
identifiers only; the paper-level quantity is $\Rinit$.}
\label{tab:theorem-traceability}
\end{table}
\FloatBarrier

\subsection{Forward Invariance}
\label{sec:forward-invariance}

By \Cref{thm:flow-bridge}, the scalar facts in this subsection apply to the
unique global SIR flow, not only to an assumed trajectory in the
\isa{SIR\_solution} locale.

The key observation is that the $I$ and $S$ equations are each homogeneous
linear in their own compartment, while $R$ has nonnegative derivative:
\begin{itemize}
\item $I' = (\beta S(t) - \gamma) \cdot I$, with coefficient $f_I(t) = \beta S(t) - \gamma$.
\item $S' = (-\beta I(t)) \cdot S$, with coefficient $f_S(t) = -\beta I(t)$.
\item $R' = \gamma I(t)$, which is nonnegative once $I \ge 0$ is established.
  This uses \Cref{lem:nonneg-deriv}, not the integrating factor.
\end{itemize}
In the checked source, names ending in \isa{nonneg} record these componentwise
forward-invariance facts, such as $I(t)\ge0$ on the interval.

\begin{theorem}[\isa{S\_nonneg}, \isa{I\_nonneg}, \isa{R\_nonneg}]
\label{thm:forward-inv}
In the \isa{SIR\_solution} locale, under its parameter, ODE, interval, and
nonnegative-initial-condition assumptions, for all $t \in [a,b]$:
$S(t) \ge 0$, $I(t) \ge 0$, and $R(t) \ge 0$.
\end{theorem}
Equivalently, for the SIR state triple, these componentwise inequalities state
forward invariance of the nonnegative orthant on the solution interval.

\noindent\emph{Remark.}
The nonnegativity proofs for $S$ and $I$ are logically independent: each uses
only differentiability, hence continuity, of the other compartment.  Neither
requires the other compartment's nonnegativity.

\begin{proof}
We first prove nonnegativity of $I$ and $S$; these two sign-preservation
arguments are order-independent because each needs only continuity of the other
compartment, not its nonnegativity.  We present $I$ first because the same
integrating-factor pattern is then reused for $S$.  The integrating-factor theorem
(\Cref{thm:linear-ode}) requires only continuity of the coefficient function
$f$ on $[a,b]$, not a sign condition on $f$.  We then prove nonnegativity of
$R$, whose monotonicity proof uses the already-proved nonnegativity of $I$.

For $I$: the ODE $I' = (\beta S - \gamma) I$ satisfies the hypotheses of
\isa{linear\_ode\_nonneg} with $f(t) = \beta S(t) - \gamma$.
Continuity of $f$ follows because differentiability of $S$ on $[a,b]$ implies
continuity there; the artifact records the resulting fact as
\isa{continuous\_S}.
No nonnegativity of $S$ is needed here; only continuity derived from differentiability.
With $I(a) \ge 0$, we get $I(t) \ge 0$.

For $S$: the ODE $S' = (-\beta I) \cdot S$ satisfies the hypotheses with
$f(t) = -\beta I(t)$, which is continuous since $I$ is continuous by
differentiability, not because $I$ has already been proved nonnegative.  With
$S(a) \ge 0$, we get $S(t) \ge 0$.

For $R$: since $I(t) \ge 0$ (just proved) and $\gamma > 0$,
we have $R'(t) = \gamma I(t) \ge 0$. With $R(a) \ge 0$, the
nondecreasing-function wrapper \isa{nonneg\_deriv\_nonneg}
(\Cref{lem:nonneg-deriv}) gives $R(t) \ge 0$.
\end{proof}

\medskip
Within \isa{SIR\_solution}, whose locale assumptions already include the
nonnegative initial conditions, the combined theorem
\isa{forward\_invariance} packages the three conclusions:
\[
  \forall t \in [a,b].\quad S(t) \ge 0,\quad I(t) \ge 0,\quad R(t) \ge 0.
\]

Unless stated otherwise, the remaining SIR theorems below are stated in the
\isa{SIR\_solution} locale and use its interval, ODE, positive-parameter, and
nonnegative-initial-value assumptions, with any additional hypotheses stated locally.

\subsection{Integrating-Factor Integral Representations}

In \isa{SIR\_solution}, the integrating factor proofs also yield integral
representations for any supplied SIR solution:

\begin{theorem}[\isa{I\_solution}]
\label{thm:i-s-solutions}
For all $t \in [a,b]$:
\[
  I(t) = I(a) \cdot \exp\!\left(\int_a^t (\beta S(s) - \gamma)\, ds\right)
\]
\end{theorem}

\begin{theorem}[\isa{S\_solution}]
For all $t \in [a,b]$:
\[
  S(t) = S(a) \cdot \exp\!\left(\int_a^t (-\beta \cdot I(s))\, ds\right)
\]
\end{theorem}

Both integral representations instantiate \Cref{thm:linear-ode} with the
coefficient functions shown in the exponents.
These are not closed-form solutions because the integrands still contain the
unknown trajectory.  They are nevertheless sufficient for the
machine-checked sign arguments: the exponential factor is always positive, so
the sign is determined by the initial value.
Epidemiologically, the $S$ representation expresses susceptible depletion as an
exponential function of the cumulative infectious load
\[
  P_I(t)=\int_a^t I(s)\,ds.
\]
The $I$ representation analogously expresses infectious growth or decline
through the cumulative excess growth rate
\[
  G_I(t)=\int_a^t(\beta S(s)-\gamma)\,ds.
\]

\subsection{Conservation of Total Population}

\begin{theorem}[\isa{conservation}]
Define $N = S(a) + I(a) + R(a)$.  For all $t \in [a,b]$:
\[
  S(t) + I(t) + R(t) = N
\]
\end{theorem}

\begin{proof}
Instantiate \Cref{thm:conservation} with $f = S$, $g = I$, and $h = R$.
The derivative witnesses are the three SIR right-hand sides.  Their zero-sum
condition is the algebraic identity
$(-\beta SI) + (\beta SI - \gamma I) + (\gamma I) = 0$, so the derivative of
$S+I+R$ is zero on the interval.  Differentiability supplies the continuity
premise for \isa{DERIV\_isconst2}, yielding constancy of the total population.
This is elementary on paper, but in the artifact it is the named zero-sum
derivative fact later reused by the compact-continuation and boundedness proofs.
\end{proof}

\subsection{Monotonicity}

\begin{theorem}[\isa{S\_nonincreasing}]
For $s, t \in [a,b]$ with $s \le t$: $S(t) \le S(s)$.
\end{theorem}

\begin{theorem}[\isa{R\_nondecreasing}]
For $s, t \in [a,b]$ with $s \le t$: $R(s) \le R(t)$.
\end{theorem}

\begin{proof}
$S' = -\beta SI \le 0$ (product of nonnegative factors, negated; using the
\emph{derived} nonnegativity of $S$ and $I$).
$R' = \gamma I \ge 0$ (product of positive $\gamma$ and nonnegative $I$).
Apply the framework's monotonicity wrappers.
\end{proof}

\subsection{The Kermack--McKendrick Phase-Plane Invariant}
\label{sec:phase-plane}

The classical Kermack--McKendrick invariant~\cite{kermack1927sir,hethcote2000mathematics}
constrains the SIR phase-plane trajectory to a level curve. It requires the
additional hypothesis $S(a) > 0$
(strict positivity of the initial susceptible population).
Epidemiologically, this is the phase-plane constraint relating susceptible
depletion to the rise and fall of the infectious compartment; the formal theorem
below proves only this level-curve constancy, not terminal behavior.
By \isa{linear\_ode\_pos} applied to $S' = (-\beta I) S$, strict positivity
propagates: $S(t) > 0$ for all $t \in [a,b]$.  This ensures $\ln(S(t))$ is
well-defined and that $\ln \circ S$ is differentiable on $[a,b]$ in the same
full-neighbourhood sense, by the chain rule for compositions with positive
argument.  The same facts give the continuity needed for the constancy argument.

\begin{definition}[\isa{KM\_invariant}]
Under the theorem's $S(a)>0$ hypothesis, the dimensionally meaningful
mathematical family is
\[
  V_{S_\ast}(t)=I(t)+S(t)-\frac{\gamma}{\beta}\ln\!\bigl(S(t)/S_\ast\bigr),
  \qquad S_\ast>0.
\]
Changing $S_\ast$ shifts the expression by an additive constant, so the
constancy theorem is unchanged.  The checked real-coordinate statement treats
compartment values as unitless coordinates, specializes this family to the
coordinate reference value $S_\ast=1$, as a coordinate convention rather than a
dimensional claim, and defines
$V(t) = I(t) + S(t) - \frac{\gamma}{\beta} \ln(S(t))$.
The invariant is the assertion that $V(t)$ is constant on $[a,b]$.
\end{definition}
Formally, $\ln$ is Isabelle's real logarithm applied to the positive real
coordinate $S(t)$.  Since $I(t)$ and $S(t)$ are added, they are understood on the
same real-valued population scale used by the unnormalized SIR system.
The fundamental checked fact is the constancy relation, not an absolute
dimensioned value of the logarithmic expression.
\Cref{sec:limitations} records the dimensional-analysis limitation.

\begin{lemma}[\isa{S\_pos\_preserved}]
If $S(a) > 0$, then $S(t) > 0$ for all $t \in [a,b]$.
\end{lemma}

\begin{proof}
By \isa{linear\_ode\_pos} applied to $S' = (-\beta I) S$.
\end{proof}

\begin{theorem}[\isa{KM\_invariant\_constant}]
\label{thm:km-invariant}
If $S(a) > 0$, then for all $s, t \in [a,b]$:
\[
  V(s) = V(t)
\]
\end{theorem}

\begin{proof}
Using $(\ln \circ S)'(t)=S'(t)/S(t)$ by the chain rule and $S(t)>0$, compute
\[
  V'(t) = I'(t) + S'(t) - \frac{\gamma}{\beta}\cdot\frac{S'(t)}{S(t)}.
\]
Substituting the ODE system and using $S(t) > 0$:
\begin{align*}
V'(t)
  &= (\beta S I - \gamma I) + (-\beta S I)
     - \frac{\gamma}{\beta}\cdot\frac{-\beta I S}{S} \\
  &= \beta S I - \gamma I - \beta S I + \gamma I \\
  &= 0.
\end{align*}
By \isa{DERIV\_isconst2}, $V$ is constant.
\end{proof}

\begin{corollary}[\isa{KM\_invariant\_value}]
If $S(a) > 0$, then for all $t \in [a,b]$:
\[
  V(t) = I(a) + S(a) - \frac{\gamma}{\beta}\ln(S(a)).
\]
\end{corollary}

Equivalently, subtracting the initial value gives the logarithmic-ratio form
\[
  I(t)+S(t)-I(a)-S(a)
  = \frac{\gamma}{\beta}\ln\!\left(\frac{S(t)}{S(a)}\right).
\]
This is often the most natural dimensional reading; the checked theorem above
remains the coordinate-level constancy statement.

The theorem states constancy symmetrically, matching the Isabelle lemma; the
corollary gives the initial-value form most commonly used in epidemiological
analysis.  Here \(\ln\) is the natural logarithm, formalized as Isabelle's real
\isa{ln}; the strict positivity of \(S(t)\) is used both to define the logarithm
and to justify cancellation by \(S(t)\) in the derivative calculation.
If $S(a)=0$, the logarithmic invariant is not stated: the $S$ integral
representation gives $S(t)=0$ on $[a,b]$, so the log-based phase-plane invariant
is unnecessary for that boundary case.

This invariant constrains the phase-plane trajectory to a one-dimensional level
set of $V$.
Geometrically, the pair $(S(t),I(t))$ lies on the level curve
\[
  \left\{(x,y)\in\R_{>0}\times\R_{\geq 0} \mid
    y + x - \frac{\gamma}{\beta}\ln x = V(a)\right\}
\]
in the $(S,I)$ phase plane.
Here $x$ represents the susceptible population and $y$ the infectious
population; forward invariance supplies the nonnegativity of~$y$ along the
trajectory.

\begin{figure}[H]
\centering
\begin{tikzpicture}[>=Stealth,scale=0.95]
  \draw[->] (0,0) -- (4.6,0) node[right] {$S$};
  \draw[->] (0,0) -- (0,3.0) node[above] {$I$};
  \draw[dashed] (1.75,0) -- (1.75,2.45);
  \node[below] at (1.75,0) {$S=\gamma/\beta$};
  \draw[thick]
    plot[smooth] coordinates {(4.0,0.45) (3.15,1.35) (2.25,2.12)
      (1.75,2.28) (1.2,1.95) (0.55,0.75)};
  \draw[thick,->]
    plot[smooth] coordinates {(3.35,1.15) (2.65,1.78) (2.0,2.18)};
  \draw[thick,->]
    plot[smooth] coordinates {(1.65,2.26) (1.25,1.97) (0.85,1.35)};
  \node[align=left] at (3.25,2.6)
    {\small level curve\\[-1pt]\small $I+S-\frac{\gamma}{\beta}\ln S=V(a)$};
\end{tikzpicture}
\caption{Illustrative schematic phase-plane sketch, drawn as a level-curve
diagram rather than generated from a checked trajectory or computed from the
formal artifact and not to scale.  The figure gives a phase-plane reading of the
Kermack--McKendrick invariant.  The sketch schematically represents a level
curve of $V$, arrows indicate increasing time (right-to-left as $S$ decreases),
and the vertical line marks the stationary-infection threshold from
\Cref{sec:R0}; the artifact does not
separately prove that every trajectory crosses the threshold.}
\label{fig:phase-plane}
\end{figure}
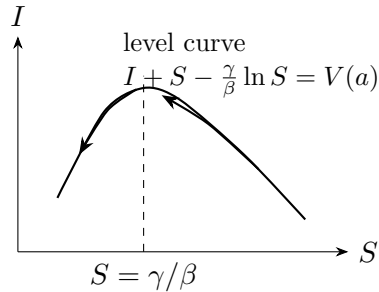

Together with conservation, it is the classical starting point for final-size
and threshold analyses of SIR
dynamics~\cite{hethcote1976qualitative,hethcote2000mathematics,diekmann2013tools};
\Cref{sec:limitations} records that the terminal final-size equation is not
formalized here.
Deriving that equation would also require terminal-limit or asymptotic
convergence facts, which are outside the present artifact.

\FloatBarrier

\subsection{Epidemic Growth Condition}
\label{sec:growth-condition}

The following algebraic step uses the derived forward invariance
$I(t)\ge0$; without that sign premise the equivalence is false over arbitrary
real values of $I(t)$.  By the SIR ODE, the left-hand expression in the theorem
is precisely the derivative $I'(t)$.  Throughout the threshold subsections,
$I'(t)$ abbreviates the derivative value supplied by the \isa{ode\_I}
assumption, namely $\beta S(t)I(t)-\gamma I(t)$.

\begin{theorem}[\isa{epidemic\_growth\_iff}]
In the \isa{SIR\_solution} locale, where $I(t)\ge0$ is derived by
\Cref{thm:forward-inv}, for $t \in [a,b]$:
\[
  \beta S(t) I(t) - \gamma I(t) > 0
  \quad\iff\quad
  I(t) > 0 \;\wedge\; \beta S(t) > \gamma
\]
\end{theorem}
The proof factors the left-hand side as
$I(t)(\beta S(t)-\gamma)$.  Since forward invariance gives $I(t)\ge0$,
strict positivity of this product is equivalent to $I(t)>0$ and
$\beta S(t)-\gamma>0$.  Naming this equivalence in Isabelle makes the
supporting algebraic factorization and positivity side condition explicit for
the reproduction-threshold theorems below.

\subsection{Initial Effective Threshold Ratio}
\label{sec:R0}

\subsubsection{Definitions and Notation}

\paragraph{Terminology.}
The checked identifier \isa{R\_zero} is historical.  In the paper we write this
quantity as $\Rinit=\beta S(a)/\gamma$, the initial effective threshold ratio
for the chosen initial state.  It is not the disease-free basic reproduction
number $\Rclassical=\beta N/\gamma$ of the fully susceptible disease-free state,
except when $S(a)=N$.
We avoid calling $\Rinit$ ``$R_0$'' because $S(a)$ need not equal the
fully susceptible disease-free population $N$.
In the normalized convention $N=1$, this fully susceptible disease-free
quantity specializes to the familiar scalar threshold $\beta/\gamma$.

\paragraph{Definitions and caveats.}
\begin{definition}[\isa{R\_zero}, \isa{R\_eff}]
\[
  \Rinit = \frac{\beta \cdot S(a)}{\gamma}, \qquad
  \Reff(t) = \frac{\beta \cdot S(t)}{\gamma}
\]
\end{definition}

Here $S$ is on the same real-valued population-count scale as
\Cref{eq:sir-s,eq:sir-i,eq:sir-r}.  For a fraction-valued formulation one would
normalize $N=1$ and rescale $\beta$ accordingly.  We call
$\beta S(t)/\gamma$ the effective threshold ratio; under this unnormalized
mass-action convention it plays the role usually played by an effective
reproduction number for the current susceptible count.  We avoid calling it an
effective reproduction number except in this deterministic mass-action threshold
sense; it is not a stochastic secondary-case count or a population-normalized
incidence convention.  The notation summary is:
\[
\begin{array}{rcll}
  \Rinit &=& \beta S(a)/\gamma &= \Reff(a)\quad\text{(this paper)},\\
  \Reff(t) &=& \beta S(t)/\gamma &\text{(time-dependent mass-action threshold)},\\
  \Rclassical &=& \beta N/\gamma &\text{(fully susceptible DFE quantity)}.
\end{array}
\]
In the theorem statements below, \(\Reff(t)\) always denotes this deterministic
mass-action threshold ratio.
Following the disease-free next-generation convention
\cite{diekmann1990definition,vandendriessche2002reproduction,diekmann2010construction}, we write
$\Rclassical$ for the mass-action SIR basic reproduction number obtained by
substituting the total population $N$ into the fully susceptible disease-free
state; thus $\Rclassical=\beta N/\gamma$, and
$\Rinit=\Rclassical$ only when $S(a)=N$.
The comparison with $\Rclassical$ is explanatory: the checked threshold theorems
below use $\Rinit$ and $\Reff(t)$, not a separate named theorem about
$\Rinit\le\Rclassical$.
For this single-infectious-compartment SIR model, at a fixed susceptible
background \(S\), and in particular at the fully susceptible DFE, the
next-generation matrix calculation reduces to the scalar ratio $\beta S/\gamma$.
The formalization treats all quantities as real values without physical-unit
tracking.  With those mass-action units,
$\Rinit=\beta S(a)/\gamma$ is dimensionless and $\beta SI$ is a population-rate
term.
The initial-threshold interpretation is most informative when $S(a)>0$; if
$S(a)=0$, then $\Rinit=0$ and no susceptible population remains for new
infections, although the algebraic theorems below still apply under their stated
hypotheses.
If $I(a)=0$, \Cref{thm:i-s-solutions} already gives $I(t)=0$ for all $t$
throughout the interval, so the infection-free trajectory is handled separately
from the positive-infection threshold statements below.  Substituting
$I(t)=0$ into the SIR ODE gives $S'=-\beta S(t)\cdot 0=0$ and
$R'=\gamma\cdot 0=0$, yielding a stationary infection-free trajectory on the
interval; the familiar fully susceptible DFE used for $\Rclassical$ is the
special case $S=N$ and $R=0$.
The vanishing-$I$ consequence is checked through \Cref{thm:i-s-solutions};
the preceding stationary-$S,R$ observation is explanatory algebra from the
displayed ODE, not a separately named Isabelle theorem.

The effective threshold ratio $\Reff(t)$ is nonincreasing because $S$ is
nonincreasing.  Therefore, if $\Rinit \le 1$, the infected compartment is
nonincreasing on the entire interval $[a,b]$.
The derivative identity behind the threshold statement is
\[
  I'(t) = I(t)\,(\beta S(t)-\gamma).
\]
The preceding subsection, \Cref{sec:growth-condition}, records the raw algebraic
factorization of this identity.

\subsubsection{Threshold Theorems}

All threshold ratios in this subsection use the unnormalized mass-action
convention from \Cref{sec:sir-ode-system}, where
$\Reff(t)=\beta S(t)/\gamma$.

\paragraph{Pointwise threshold theorems.}
\begin{theorem}[\isa{epidemic\_growth\_R\_eff}]
For $t \in [a,b]$ with $I(t)>0$:
\[
  I'(t) > 0 \;\Longleftrightarrow\; \Reff(t) > 1.
\]
\end{theorem}

\begin{theorem}[\isa{initial\_epidemic\_growth}]
If $I(a) > 0$ and $\Rinit > 1$, then
$I'(a)=I(a)(\beta S(a)-\gamma)>0$ (infectious compartment initially increasing).
\end{theorem}
This is a pointwise initial-growth statement only.  Since $S$ is
nonincreasing, $\Rinit>1$ alone does not imply that $I'$ remains positive
throughout the whole interval.

\begin{theorem}[\isa{initial\_epidemic\_decline}]
If $I(a) > 0$ and $\Rinit < 1$, then
$I'(a)=I(a)(\beta S(a)-\gamma)<0$ (infectious compartment initially decreasing).
\end{theorem}

\begin{theorem}[\isa{initial\_no\_epidemic}; initial non-growth]
If $I(a) > 0$ and $\Rinit \le 1$, then
$I'(a)=I(a)(\beta S(a)-\gamma)\le 0$.
\end{theorem}
Despite the historical Isabelle name \isa{initial\_no\_epidemic}, this is an
initial non-growth theorem; it does not assert that no infected individuals are
present or that no epidemic interpretation is possible.

\paragraph{Interval-wide monotonicity under $\Rinit \le 1$.}
The pointwise condition strengthens to an interval-wide monotonicity result:

\begin{theorem}[\isa{I\_nonincreasing\_if\_R\_zero\_le\_one}]
If $\Rinit \le 1$, then for all $s, t \in [a,b]$ with $s \le t$: $I(t) \le I(s)$.
\end{theorem}

\begin{proof}
Since $S$ is nonincreasing, $\beta>0$, and
$\Rinit = \beta S(a)/\gamma \le 1$, we have
$S(t)\le S(a)$ and hence
$\beta S(t) \le \beta S(a) \le \gamma$ for all $t$. Thus
forward invariance gives $I(t)\ge0$, while
$\beta S(t)-\gamma\le0$, so
$I'(t) = I(t)(\beta S(t) - \gamma) \le 0$, and $I$ is nonincreasing.
\end{proof}

This is the key formalized epidemiological consequence connecting the initial effective
threshold to infectious-compartment monotonicity
on $[a,b]$:
when $\Rinit \le 1$, the infection is monotonically nonincreasing on the entire interval.
Together, the growth, decline, non-growth, and interval-wide monotonicity lemmas
make explicit which conclusions are pointwise and which are global on $[a,b]$.

\subsection{Stationary Infection Condition}

\paragraph{Stationary condition.}
The checked source name \isa{peak\_iff} is historical; the paper-level theorem is
only the stationary condition below, not a peak-existence theorem.
\begin{theorem}[Zero-derivative characterization for \(I\); source name \isa{peak\_iff}]
For $t \in [a,b]$ with $I(t) > 0$:
\[
  \beta S(t) I(t) - \gamma I(t) = 0
  \quad\iff\quad
  S(t) = \gamma / \beta
\]
\end{theorem}

This characterizes points where $I'=0$ (necessary condition for an interior
local extremum of $I$). It does not by itself prove that such a point exists or
that it is a global maximum.  Algebraically, while $I>0$, the sign condition for
$I'$ is consistent with a transition from growth to non-growth if $S$ crosses
$\gamma/\beta$ from above; the artifact does not separately prove an existence
or crossing theorem for such a point (\Cref{sec:limitations}).

\subsection{Boundedness}

\begin{theorem}[\isa{S\_bounded}, \isa{I\_bounded}, \isa{R\_bounded}, \isa{compartment\_bounds}]
For all $t \in [a,b]$: $0 \le S(t) \le N$, $0 \le I(t) \le N$, and $0 \le R(t) \le N$,
where $N = S(a) + I(a) + R(a)$ is nonnegative by the initial-value assumptions.
\end{theorem}

\begin{proof}
The initial assumptions give $N\ge0$ because it is a sum of nonnegative initial
compartment values.
Lower bounds are the \emph{derived} forward invariance (\Cref{thm:forward-inv}).
For the upper bounds, conservation gives $S(t)+I(t)+R(t)=N$.  Hence
$S(t)\le N$ because $I(t),R(t)\ge0$; $I(t)\le N$ because $S(t),R(t)\ge0$; and
$R(t)\le N$ because $S(t),I(t)\ge0$.
\end{proof}

\section{Related Work}
\label{sec:related-work}

\subsection{ODE Formalization in Isabelle}

Immler and H\"olzl~\cite{immler2012ode,afp-ode} formalized the
Picard-Lindel\"of existence and uniqueness theorem in Isabelle, building a
library for initial value problems; this paper uses the AFP~2024 version of
that library, including its local-flow and continuation infrastructure. This was
later extended to support the
Poincar\'e--Bendixson
theorem~\cite{immler2020poincare}, verified reachability
analysis~\cite{immler2015reachability}, the Lorenz
attractor~\cite{immler2018lorenz}, and other dynamical-systems results.
The Poincar\'e--Bendixson formalization addresses planar limit-set structure;
although SIR trajectories project to a planar $(S,I)$ phase plane, the present
SIR proof instead needs finite-interval sign facts and compact continuation to
obtain global forward existence rather than limit-set classification.

Our work builds directly on this library for the SIR instantiation: we prove
continuous differentiability of the polynomial SIR vector field, obtain the
local flow, and then prove global forward existence for nonnegative initial
data.  The paper does not contribute a new Picard--Lindel\"of theorem or a new
numerical ODE solver.  Its contribution is the domain layer above that
infrastructure: reusable sign-preservation and conservation lemmas, a
locale-based SIR instantiation, and a bridge showing that the qualitative
theorems apply to the unique AFP flow.
This follows a common proof-assistant contribution pattern: classical
mathematics becomes new infrastructure when mechanized with reusable interfaces
and explicit side conditions, as in the Isabelle ODE, reachability, and
dynamical-systems developments cited above.
Compared with those broader dynamics formalizations, this paper is deliberately
narrower in mathematical scope: it is a single epidemiological ODE case study.
The contribution is the end-to-end domain theorem package, reusable scalar
interfaces, and non-circular transfer from the constructed AFP flow to the
qualitative SIR facts.

\subsection{Verified Reachability and Hybrid-System Verification}

Immler~\cite{immler2015reachability} demonstrates verified reachability analysis
of continuous systems in Isabelle using rigorous numerics (affine arithmetic and
adaptive Runge--Kutta methods), yielding certified over-approximations of
reachable sets.

Our approach is \emph{qualitative}: we prove structural properties
(conservation, forward invariance, phase-plane structure, threshold behavior)
without computing trajectories numerically. The two approaches address different
verification goals---numerical prediction versus structural understanding---and
are complementary.
The epidemiological facts themselves are classical; the novelty claimed here is
their checked Isabelle/AFP certification, reusable interface packaging, and
non-circular connection to the constructed flow.

Hybrid-system verification systems such as differential dynamic
logic~\cite{platzer2008dl,platzer2018logical} and KeYmaera
X~\cite{fulton2015keymaerax} also reason about differential equations, including
safety and invariant properties of continuous dynamics via differential-invariant
reasoning.  Those tools target broad hybrid-system proof automation over hybrid
programs and specialized proof calculi, whereas this work develops standard
Isabelle/HOL theorems about a concrete epidemiological ODE and its reusable
scalar proof patterns.
The nonnegative orthant and conservation simplex used for SIR can also be
expressed as differential invariants; our artifact proves them instead through
integrating factors, scalar monotonicity, and AFP flow continuation.  This route
avoids the tangent-cone infrastructure a Nagumo-style
proof~\cite{nagumo1942integralkurven} would require, stays inside standard
HOL-Analysis, and yields an integral representation, not only a sign conclusion.
Tan and Platzer's axiomatic treatment of existence and liveness for
differential equations~\cite{tan2021existence} is especially relevant on the
hybrid-systems side because it isolates subtleties such as finite-time escape.
Its setting is differential dynamic logic proof theory, whereas our contribution
is a concrete Isabelle/HOL/AFP proof that the SIR flow exists globally and
supports the stated epidemiological consequences.
Fulton's web note on disease-spread models in KeYmaera
X~\cite{fulton2020sirkeymaera} uses a simplified SIR model to explore a
hospital-capacity reachability property during the COVID-19 pandemic.  It
further illustrates the hybrid-system safety style, but it is not a
proof-assistant formalization of the Picard--Lindel\"of/global-existence,
phase-plane, and threshold theorem package studied here.
We do not claim that Isabelle/HOL is the shortest way to prove the safety
invariants alone: for example, a hybrid-system prover can express the
nonnegative orthant and conservation simplex as differential invariants.  The
choice of Isabelle/HOL is driven by the desired artifact shape: ordinary HOL
theorems about the AFP Picard--Lindel\"of flow, reusable scalar integral
representations, and theorem names that later Isabelle developments can import
directly.

Other proof assistants provide relevant analysis infrastructure.  Coquelicot
extends Coq's real-analysis library~\cite{boldo2015coquelicot}, and Makarov and
Spitters formalized a constructive Picard algorithm for ODEs in
Coq~\cite{makarov2013picard}.  Park and Thies formalize Taylor models and power
series in Coq for certified exact-real computation of ODE
solutions~\cite{park2024coq}.  HOL Light has formalized substantial
Euclidean-space foundations~\cite{harrison2013hol-light-euclidean} and
metric-space foundations including the Banach fixed-point theorem~\cite{maggesi2018metric};
Lean's mathlib is a broad mathematical library~\cite{mathlib2020}; current
mathlib4 documentation also includes a Picard--Lindel\"of ODE module
\isatt{Mathlib.Analysis.ODE.PicardLindelof}.\footnote{\url{https://leanprover-community.github.io/mathlib4_docs/Mathlib/Analysis/ODE/PicardLindelof.html},
last checked \lastcheckeddate.}  These libraries are adjacent
foundations rather than comparable SIR developments: they provide analysis
infrastructure on which one could build domain-specific compartmental proofs.
The targeted positioning search reported in \Cref{sec:formal-epidemiology} did
not find a SIR or compartmental-model formalization built on the Lean ODE
module.  It also records how we checked for a closer deterministic SIR
proof-assistant artifact.

Outside interactive theorem proving, formal-methods work in computational
biology has developed executable and process-algebraic modeling styles, such as
executable cell biology~\cite{fisher2007executable} and
Bio-PEPA~\cite{ciocchetta2009biopepa}.  Those approaches target simulation,
model checking, or analysis of biological networks rather than HOL proofs of
deterministic ODE theorems, but they help situate this work within a broader
formal-methods ecosystem for biological models.

\subsection{Formal Epidemiology in Other Proof Assistants}
\label{sec:formal-epidemiology}

To the best of our knowledge, after a targeted positioning search of the Archive
of Formal Proofs entry catalogue, Lean/mathlib references, the Coq opam package
index, and public web-search results for the queries below, we are not aware of a published proof-assistant
artifact that combines Picard--Lindel\"of instantiation, non-circular global
existence, forward invariance, phase-plane invariance, and threshold properties
for the deterministic SIR ODE (last checked \lastcheckeddate).\footnote{Representative
search queries included
``SIR epidemic Coq proof assistant'', ``SIR ordinary differential equation
Isabelle'', ``compartmental model Lean theorem prover'', and ``epidemic ordinary
differential equation proof assistant''; returned records were screened for
proof-assistant artifacts addressing deterministic SIR ODE theorem packages.}
This is scoped positioning evidence, not a systematic literature review or a
proof of absence: private, unpublished,
differently indexed, or other-prover developments could exist.  Accordingly, the
contribution claim does not depend on absolute historical priority for all
mechanized SIR reasoning; it is the submitted Isabelle/AFP bridge, reusable
scalar interface, and theorem package that are evaluated here.

H\"olzl~\cite{holzl2012markov} formalized discrete-time Markov chains in
Isabelle (AFP entry).  This is relevant because stochastic epidemic models are
often formulated as Markov chains, but it does not address the deterministic
ODE setting studied here.
The targeted positioning search above did not identify a proof-assistant
stochastic SIR formalization that changes the deterministic-ODE comparison;
stochastic models are outside the present artifact's scope.

\subsection{Positioning}

\Cref{tab:positioning} summarizes what our development reuses from the
Isabelle library and what is formalized in this artifact.  The artifact
organization and size are summarized separately in
\Cref{tab:artifact-architecture}.

\begin{table}[ht]
\centering
\small
\begin{tabular}{@{}>{\raggedright\arraybackslash}p{0.32\linewidth}
                >{\raggedright\arraybackslash}p{0.22\linewidth}
                >{\raggedright\arraybackslash}p{0.34\linewidth}@{}}
\toprule
\textbf{Component} & \textbf{Source} & \textbf{Status} \\
\midrule
FTC, exp, continuous\_on & HOL-Analysis & Reused (HOL-Analysis) \\
Derivative rules & HOL-Analysis & Reused (HOL-Analysis) \\
Locales, type classes & Isabelle/Pure & Reused (Isabelle/Pure) \\
\midrule
Integrating factor lemma & This work & Formalized here (generic) \\
Sign preservation (nonneg/pos) & This work & Formalized here (generic) \\
Monotonicity wrappers & This work & Formalized here (interface) \\
3-compartment conservation & This work & Formalized here (generic) \\
\midrule
SIR Picard--Lindel\"of instantiation & This work + AFP ODE & Formalized here (SIR-specific) \\
Global forward existence & This work & Formalized here (SIR-specific) \\
SIR forward invariance & This work & Formalized here (SIR-specific transfer) \\
Kermack--McKendrick invariant & This work & Formalized here (SIR-specific) \\
$\Rinit$, $\Reff$, and threshold conditions & This work & Formalized here (SIR-specific) \\
\midrule
Picard--Lindel\"of theorem & Immler--H\"olzl (AFP ODE) & Reused (AFP ODE) \\
\bottomrule
\end{tabular}
\caption{Provenance of formal components.  The comparison is based on the
  targeted positioning search in \Cref{sec:related-work}; ``generic'' results are
  reusable for the stated scalar or three-compartment shape, while
  ``SIR-specific'' results depend on the SIR vector field and its flow bridge.}
\label{tab:positioning}
\end{table}

\noindent Within the targeted positioning search above, the distinguishing features of this
development are the combination already summarized in \Cref{sec:introduction}:
AFP Picard--Lindel\"of instantiation for SIR, global forward existence under
nonnegative initial conditions, a reusable scalar framework, the
Kermack--McKendrick invariant and threshold facts, and a locale architecture that
separates generic compartmental reasoning from model-specific instantiation.
The main formal difficulty is connecting scalar qualitative proofs to a globally
defined AFP flow without circularly assuming the forward invariance used to
obtain global existence.

\section{Discussion}
\label{sec:discussion}

\subsection{What Is Assumed vs. Proved}

The scalar locale \isa{SIR\_solution} separates assumptions from consequences;
\Cref{tab:assumed-proved} summarizes the results detailed in \Cref{sec:sir-model}.
\begin{table}[H]
\centering
\small
\begin{tabular}{>{\raggedright\arraybackslash}p{0.42\linewidth}
                >{\raggedright\arraybackslash}p{0.48\linewidth}}
\toprule
\textbf{Assumed} & \textbf{Proved} \\
\midrule
Differentiable trajectories $S,I,R$ satisfying the SIR ODE on $[a,b]$; parameters
$\beta,\gamma>0$; nonnegative initial conditions $S(a),I(a),R(a)\ge 0$.
&
Forward invariance, integral representations for $S$ and $I$, conservation,
monotonicity, boundedness, the phase-plane invariant under $S(a)>0$, and
threshold facts under their stated positivity hypotheses. \\
\bottomrule
\end{tabular}
\caption{Locale-level assumptions and derived qualitative facts.}
\label{tab:assumed-proved}
\end{table}

The locale \isa{SIR\_ODE} discharges the trajectory assumption for the standard
forward initial-value problem by constructing the unique global SIR flow from
positive parameters and nonnegative initial values.

\subsection{The Integrating Factor vs. Picard--Lindel\"of}
\label{sec:discussion-if-vs-pl}

Classical textbook treatments of SIR dynamics often use boundary-repulsion
arguments for forward invariance.  The integrating-factor route used here gives
more than a sign conclusion: on any interval where the scalar ODE is satisfied,
it produces an integral representation for $X(t)$ in terms of $X(a)$ and an
exponential integral, without invoking uniqueness.  This proof is reusable
independently of the Picard--Lindel\"of infrastructure, but it is limited to the
homogeneous-linear shape discussed in \Cref{sec:framework}.  The derivative
endpoint convention behind the scalar framework is stated in \Cref{sec:derivatives}.
The separation keeps scalar sign and conservation lemmas importable without
coupling them to the SIR-specific existence proof.

\subsection{Proof Engineering Observations}

\paragraph{Targeted derivative rules.}
\Cref{sec:framework} records the local proof details for the integrating-factor
lemma.  The broader engineering lesson is that derivative automation was used
selectively: specific product and chain rules were applied explicitly rather
than asking Isabelle to search the full derivative-introduction rule set.

\paragraph{Continuous composition.}
Continuity goals are kept close to their mathematical shape by applying
\isa{continuous\_on\_exp} and \isa{continuous\_on\_ln} directly to already-proved
continuity facts for the inner functions.

\paragraph{The FTC interface.}
The endpoint distinction remains explicit: FTC lemmas first produce derivatives
within $[a,b]$, whereas the scalar locale assumes full-neighbourhood
derivatives.
Together, at the pinned commit, these proof-engineering patterns keep the scalar
\path{theories/SIR/} development's 808 raw theory lines, excluding the
1,021-line existence bridge in \path{work/SIR_Existence.thy} and
\path{work/SIR_Main.thy}, focused on model-specific obligations rather than
repeated calculus-library instantiation boilerplate, so the model-specific layer
remains small enough to audit against the paper's
theorem summaries.

\subsection{Scope Boundaries and Limitations}
\label{sec:limitations}

The following are deliberate scope boundaries rather than gaps in the stated
results; they do not invalidate the results proved within the present
contribution.
\begin{itemize}
\item \emph{Model scope.}  We formalize only the basic closed-population SIR
  model without births, deaths, vital dynamics, vaccination, age structure, or
  spatial heterogeneity.  The checked SIR vector field is autonomous with
  constant parameters; models with seasonal forcing or other explicitly
  time-dependent rates would require a separate nonautonomous instantiation.
\item \emph{Analysis scope.}  We do not formalize equilibria, local or global
  stability, asymptotic convergence, or the final-size relation (implicit
  equation for terminal susceptible fraction).  The artifact does not establish
  existence or uniqueness of an infection peak, a crossing time at which
  $S(t)=\gamma/\beta$, or the full epidemic-curve shape for $\Rinit>1$, such as
  initial growth followed by eventual decline.  The checked theorem
  \isa{peak\_iff} records a stationary condition at such a point when it exists;
  stronger peak or curve-shape results require sign-change, crossing-time,
  terminal-limit, or asymptotic arguments beyond the current artifact.  For the
  present closed-population SIR model without vital dynamics there is no positive
  endemic equilibrium; endemic-equilibrium analysis belongs to model extensions
  with replenishment of susceptibles.
\item \emph{Formalization scope.}  The scalar framework theorems are
  finite-interval results; global forward existence is proved separately in the
  SIR-specific \isa{SIR\_ODE} locale by combining the AFP flow construction with
  conservation, nonnegativity, and compact-continuation infrastructure.  The
  Isabelle development treats compartment values and parameters as real numbers
  and does not formalize physical units or dimensional analysis, including
  dimensional consistency of logarithmic expressions such as the
  Kermack--McKendrick invariant.
\item \emph{Framework generality.}  The present paper gives one full checked
  model instantiation, SIR.  The passing mention of SI reuse after
  \Cref{cor:nonneg} is only a possible target, not a second mechanized model.
  The framework operates on scalar compartmental equations, not on the ODE system
  as a whole; system-level properties such as Jacobian analysis are outside the
  current scope.  The conservation theorem handles exactly three compartments;
  an $n$-compartment generalization would use indexed families and connect
  finite-sum derivative rules and Isabelle's \isa{Big\_Operators}
  infrastructure to the zero-sum derivative hypothesis.  The framework also does
  not handle affine compartment equations (ODEs with additive source terms, such
  as $E' = \beta S I - \sigma E$ in SEIR); as noted in \Cref{sec:framework},
  extending the scalar lemma via variation of parameters is a natural next step.
  Scalability to larger compartment systems is therefore an engineering question
  about the needed indexed interfaces, locale interpretations, and
  model-specific sign and conservation obligations; this artifact does not
  measure proof-search or locale-interpretation scaling for five-or-more
  compartment models.
\item \emph{Computational scope.}  The artifact is theorem infrastructure, not
  an executable epidemiological simulator; it does not use Isabelle code
  generation to compute trajectories.
\end{itemize}

\subsection{Artifact Statement}
\label{sec:artifact}

The complete Isabelle/HOL formalization is maintained at:
\begin{center}
\url{https://github.com/d0d1/verified-compartmental-models}
\end{center}
This is the Isabelle artifact repository; the manuscript sources are maintained
separately, and the commit below identifies the artifact snapshot checked by the
paper.
The results reported here correspond to commit
\begin{center}
\small\texttt{c99b9865041414c8ba9f3601cb335433a4729ce1}.
\end{center}
To reproduce the checked artifact, clone the repository, check out that commit,
install Isabelle~2024 and AFP~2024, and from inside the checked-out artifact
repository build the session with
\begin{center}
\small\texttt{isabelle build -d <AFP>/thys -d . Verified\_Compartmental\_Models}.
\end{center}
Here \texttt{<AFP>} denotes the root directory of the local AFP~2024 checkout.
The repository is currently private; for journal review, the submitted materials
must include the commit as a repository snapshot or make it accessible to
referees.

The artifact repository contains a \path{ROOT} file and a \path{README.md}
recording the Isabelle/AFP dependency versions and build instructions, and the
snapshot can be built with the command above without changing the checked
sources.
The formalization depends on AFP~2024's
\path{Ordinary_Differential_Equations} entry~\cite{afp-ode}.  The artifact is
distributed under the BSD-3-Clause license.  The pinning is intentional: later
Isabelle/AFP releases may require proof-maintenance edits, so the validation
claim is for Isabelle~2024, AFP~2024, and the commit above rather than for all
future tool versions.

The reported validation run used Isabelle~2024 on Ubuntu 24.04.4 LTS (Linux 6.8,
x86-64) on a virtualized 4-vCPU AMD EPYC machine with 7.8 GiB RAM and 15 GiB
swap.  This is the recorded validation environment, not a claimed minimum
requirement.  For the pinned artifact snapshot, the build for the commit above
completed in approximately 34 seconds of wall-clock time once prerequisite
Isabelle/AFP sessions were available, and the checked
Isabelle theory files at that commit contain zero \isa{sorry} proof placeholders
and zero \isa{oops} proof-abort commands in the following source-level audit:
\begin{quote}
\footnotesize
\begin{verbatim}
git grep -nE '\b(sorry|oops)\b' c99b9865041414c8ba9f3601cb335433a4729ce1 -- '*.thy'
\end{verbatim}
\end{quote}
A companion source-level scan of the same pinned theory files found no
declarations using \isa{axiomatization}, \isa{axioms}, or \isa{oracle}:
\begin{quote}
\footnotesize
\begin{verbatim}
git grep -nE '\b(axiomatization|axioms|oracle)\b' \
  c99b9865041414c8ba9f3601cb335433a4729ce1 -- '*.thy'
\end{verbatim}
\end{quote}
The theorem families reported in the manuscript match the selected
standard qualitative SIR facts cited in \Cref{sec:introduction,sec:related-work},
including the Kermack--McKendrick invariant, monotonicity, threshold behavior,
and boundedness.  The trust base is the standard Isabelle/HOL kernel and HOL
axioms together with the imported Isabelle/AFP libraries; this artifact does not
independently verify Isabelle's implementation or the operating-system
toolchain.

For that pinned commit, the source is organized as follows.  The raw theory-line
counts are reproducible \texttt{wc -l} counts, including comments and blanks,
rather than source-lines-of-code counts:
\begin{itemize}
\item \path{theories/Framework/Compartmental_Model.thy}: generic framework
  (220 raw theory lines; 7 named theorem, lemma, or corollary declarations).
\item \path{theories/SIR/}: scalar SIR theories
  (8 theory files; 808 raw theory lines; 45 named theorem, lemma, or corollary
  declarations, excluding definitions).
\item \path{work/SIR_Existence.thy}: Picard--Lindel\"of flow construction and
  global forward bridge, with \path{work/SIR_Main.thy} as the entry point
  (2 theory files; 1,021 raw theory lines; 54 named theorem, lemma, or
  corollary declarations, excluding definitions).
\end{itemize}
Together with the single framework theory, these groups account for all 11
theory files in the checked artifact.
The \path{work/} directory name is historical; these files are committed,
checked session sources for the \isa{Verified\_Compartmental\_Models} build, not
generated scratch files.

\subsection{Using the Theorem Package}
\label{sec:using-theorem-package}

A downstream Isabelle development can import the SIR session and instantiate
\isa{SIR\_ODE} with positive parameters and nonnegative initial values to obtain
the unique global AFP flow.  From that locale, \isa{sir\_solution\_on\_interval}
packages the flow components as a \isa{SIR\_solution} instance on any finite
interval $[0,b]$, making the scalar theorems available for the constructed
trajectory.  For example, a later model proof that needs nonnegative
compartments, conservation, or monotonicity of $I$ under $\Rinit\le1$ can cite
the named SIR theorems instead of reproving those side conditions from the ODE.

\subsection{Future Work}

Some scope boundaries above are natural future-work extensions for the same
architecture.  The list is ordered from proof-engineering extensions nearest to
the present architecture toward larger model and analysis extensions.  The
nearest are the first two items: indexed conservation for $n$-compartment models
and variation of parameters for affine source terms.
\begin{enumerate}
\item \textbf{$n$-compartment conservation}: use finite sums and indexed
  families to support SEIR, SEIRS, and other multi-compartment models.
\item \textbf{Affine compartment equations}: extend the integrating-factor
  library with a variation-of-parameters theorem for equations of the form
  $X' = f(t)X + g(t)$.
\item \textbf{Final size relation}: derive the implicit equation for the
  terminal susceptible fraction using the Kermack--McKendrick invariant.
\item \textbf{Equilibria and stability}: formalize the disease-free equilibrium
  family for closed SIR, and formalize endemic equilibria, local/global
  stability, and asymptotic convergence in vital-dynamics extensions.
\item \textbf{Stochastic extensions}: connect to H\"olzl's Markov chain
  formalization~\cite{holzl2012markov} for comparison between deterministic and stochastic models.
\item \textbf{Numerical verification}: combine with Immler's
  reachability tools for rigorous trajectory bounds.
\end{enumerate}
The final-size and stability items require terminal-limit and asymptotic
arguments beyond the finite-interval scalar layer used in this paper, so they are
larger analysis extensions rather than routine theorem exports.

\section{Conclusion}
\label{sec:conclusion}

This case study gives a reusable route from proof-assistant existence
infrastructure to qualitative compartmental-ODE reasoning.  For SIR, the artifact
connects the selected textbook facts to the unique globally existing forward flow
constructed from AFP's ODE infrastructure.  The reusable layer packages
integrating-factor sign preservation, derivative-sign monotonicity, and
three-compartment conservation; the SIR-specific layer proves global forward
existence by establishing sign and conservation facts on local AFP-flow prefixes
before invoking compact continuation, then transfers the scalar theorems to that
flow.
The result is a checked, importable Isabelle theorem package for finite-interval
qualitative SIR reasoning connected to a globally existing AFP flow.

The formalization makes two textbook elisions explicit: sign and conservation
arguments must be available on existing local prefixes before compact
continuation can globalize the flow, and derivative/FTC interfaces must match
across full-neighbourhood and interval-relative statements.  The practical lesson
is to prove local scalar invariants for the constructed flow before compact
continuation, then transfer the resulting global flow back to reusable theorem
contexts.  The value of the case study is therefore not only the certified SIR
theorem set, but a reusable bridge pattern for future compartmental ODE
formalizations: it removes a layer of informal side-condition assumptions by
turning nonnegativity, conservation, and threshold facts into importable theorems
for the constructed flow.
In one sentence, the contribution is not a new SIR fact, but a transferable
pattern for safely moving from local AFP flows to globally valid compartmental
theorem packages.

\section*{Acknowledgments}
The authors thank the Isabelle/HOL and AFP maintainers for the HOL-Analysis and
ODE infrastructure used by this formalization.

\bibliographystyle{abbrv}
\bibliography{references}

\end{document}